\documentclass[11pt,a4paper]{article}
\usepackage[]{geometry}
\usepackage[mathcal]{euscript}
\usepackage{bm,url}                     
\usepackage{amsfonts}
\usepackage{amssymb}
\usepackage{amsmath}
\usepackage{amsthm}
\usepackage[utf8x]{inputenc}
\usepackage{float}
\usepackage[pdftex]{graphicx}
\DeclareGraphicsExtensions{.bmp,.png,.pdf,.jpg,.ps}
\usepackage{colortbl}  
\usepackage{todonotes}             
\usepackage{booktabs}
\usepackage{amsmath,amsfonts,amssymb,amsthm,booktabs,color,epsfig,graphicx,hyperref,url}

\newtheorem{Proposition}{Proposition}

\usepackage[english]{babel}
\usepackage{amsmath}
\usepackage{amsfonts}
\usepackage{amssymb}
\usepackage{graphicx}
\usepackage{color}
\usepackage{array}
\usepackage{mathrsfs}
\usepackage{hyperref}

\definecolor{violet}{rgb}{0.7,0,0.6}
\definecolor{OliveGreen}{RGB}{85,107,47}

\title{An Independence Test Based on Recurrence Rates. An empirical study and applications to real data}
\author{Juan Kalemkerian \\
Universidad de la República, Facultad de Ciencias\\ \\
Diego Fern\'andez\\
Universidad de la República, Facultad de Ciencias Económicas \\ y Administración}
\begin{document}
\maketitle

\begin{abstract}
In this paper we  propose several variants to perform the independence test between two random elements
based on recurrence rates. We will show how to calculate the test statistic in each one of
these cases.
From simulations we obtain that in high dimension, our test  clearly outperforms, in almost all cases, the other
widely used competitors.  
The test was performed on two data sets including small and large sample sizes and 
we show that in both cases the application of the test allows us to obtain interesting conclusions.

\end{abstract}


  \section{Introduction} 
   Let $\left( X_{1},Y_{1}\right) ,\left( X_{2},Y_{2}\right) ,...,\left(
X_{n},Y_{n}\right) $ be an i.i.d. sample of $\left( X,Y\right) ,$ $X\in S_{X}$ and 
$Y\in S_{Y}$, where $S_{X}$ and $S_{Y}$ are metric spaces. Consider the
hypothesis $H_{0}$ which asserts that $X$ and $Y$ are independent random
elements: this is the so called independence test. Independence
tests were developed at first for the case of $S_{X}=S_{Y}=\mathbb{R}$, based on the pioneering work of Galton 
\cite{Galton} and Pearson
\cite{Pearson} (this is the famous correlation test, which is widely used today). The limitations of
this hypothesis test are well known and have motivated several different proposals
in this topic, such as the classical rank test (e.g. Spearman, \cite{Spearman}, Kendall, 
\cite{Kendall}, and Blomqvist, \cite{Blomqvist}).
The independence of random
vectors was addressed for the first time in Wilks \cite{Wilks}. 
Several independence tests between random vectors have been proposed by now, see, for example,
 \cite{Genest}, \cite{Kojadinovic}, \cite{Bilodeau} and \cite{Beran}. Gretton et al. \cite{Gretton} propose a universally
consistent test based on Hilbert--Schmidt norms.  This test has become very popular, and 
works well for random vectors in high dimensions. Another 
consistent and very popular test for random vectors is
proposed by Sz\'{e}kely et al. \cite{Szekely,Szekely2}, which defines the concept of distance
covariance. It has been used
and has had a considerable impact from the moment that it was proposed. This test has been adapted 
to work in high dimensions.
Not many independence test have been developed for random elements where $X$ and $Y$ lie in an arbitrary metric space.
Recently, in \cite{Fernandez}, an independence test based on recurrence rates was proposed.
This test is based on the following simple idea: if $X$ and $Y$ are independent, then 
$d_X(X_1,X_2)$ and $d_Y(Y_1,Y_2)$ are independent for all i.i.d. $(X_1,Y_1),(X_2,Y_2)$ with the same distribution 
as $(X,Y)$, where $d_X$ and $d_Y$ are the distance functions between elements of $S_X$ and $S_Y$, 
respectively. The authors proposed working with an $L^{2}$-Cram\'er--von Mises functional.
On the one hand, from the theoretical point of view, this test is interesting because it does not need any assumptions
on the topological structure of the metric spaces (e.g. assuming that they are Banach spaces),
and in \cite{Fernandez} can be found the first study where the mathematical properties of the
recurrence rates were established.
Marwan \cite{Marwan} gives a historical review
of recurrence plots techniques, together with everything developed from them.
However, the potential of these techniques has not \ yet been studied in depth from
the point of view of  mathematical statistics.
On the other hand, from the practical point of view, this test is interesting because it can be used
for $X$ and $Y$ lying in spaces of any dimension, and in \cite{Fernandez} a power
comparison shows the very good performance of this test compared to others widely used for 
random variables and random vectors.

In the present paper we will show, using a power comparison, that the 
independence test of recurrence rates in high dimension outperforms the other competing tests
in almost all cases,
and we will study the incidence of the distance functions considered ($d_X$ and $d_Y$)
in the performance of the test. As expected, we will show that the test statistic  in high dimension
has some 
sensitivity to the choice of the distance function, $d_X$ or $d_Y$.
Also, we will propose and compare other functionals to be
taken into account, such as an $L^{1}$-Cram\'er--von Mises functional and a Kolmogorov--Smirnov functional,
and we will show how to compute the statistic in each case.
Lastly, we will present  applications of the test  to two real data sets.
  
  The rest of this paper is organized as follows. In Section 2 we define the procedure to test $H_0$ vs $H_1$ proposed in \cite{Fernandez} and we present
  the definition of the test statistics based on a functional of the $L^{1}$-Cram\'er--von Mises type, and 
  the 
  statistics based on a functional of the Kolmogorov--Smirnov type, and three different distances to use in 
  $S_X$
  and $S_Y$. In Section 3, we show how to compute the different statistics presented in Section 2.
  In Section 4 we present a simulation study that compares, under several alternatives,
  the powers of the different tests of recurrence rates when varying the test statistics and the distance functions.
  In Section 5, we compare the performance of the recurrence test of independence with others in high dimension 
  and we show that our test clearly outperforms the rest in almost all cases.
  In Section 6, we present two applications of the recurrence test of independence to meteorological data,
  one of them with small sample size and the second on involving a huge data set, and we show the ability 
  of the recurrence rate test to obtain interesting conclusions. Some concluding remarks are given in Section 7
  and the proof of the validity of the formulas established in Section 3 can be found in Section 8.

  \section{Test Approach and Different Statistics to Consider}
  Let  $\left( X_{1},Y_{1}\right) ,\left( X_{2},Y_{2}\right) ,...,\left(
X_{n},Y_{n}\right) $ be i.i.d.  samples of $\left( X,Y\right) $ where $X\in S_X,$ $Y\in S_Y$, 
 $S_X$ and $S_Y$ are metric spaces, and suppose given $r,s>0.$ To simplify the notation and without 
risk of confusion, we will use the same letter $d$ for the distance function on both
metric spaces, $S_X$ and $S_Y$. 

We define the recurrence rate for the sample of $X$ and $Y$ as%
\[
RR_{n}^{X}(r):=\frac{1}{n^{2}-n}\sum_{i\neq j}\mathbf{1}_{\left\{ d \left ( 
X_{i},X_{j} \right  )<r\right\} }, \ \ \ \ \ 
RR_{n}^{Y}(s):=\frac{1}{n^{2}-n}\sum_{i\neq j}\mathbf{1}_{\left\{ d \left ( Y_i ,Y_j \right ) <s\right\} }, 
\]%
respectively, and the joint recurrence rate for $\left( X,Y\right) $ as
\[RR_{ n}^{X,Y}(r,s):=\frac{1}{n^{2}-n}\sum_{i\neq j}\mathbf{1}_{\left\{
d \left ( X_{i},X_{j} \right ) <r\text{ },\text{ }d \left ( Y_i ,Y_j \right ) <s \right\} }.\]
If we define $p_{X}(r):=P\left( d \left ( 
X_{1},X_{2} \right  )<r\right) $ the probability that the distance
between any two elements of the sample $X$ is less than $r$ and  $p_{X,Y}(r,s):=P\left( d \left ( 
X_{1},X_{2} \right  )<r,\text{ }d \left ( Y_{1},Y_{2} \right )
<s \right) $ the joint probability that the distance between any two elements of 
the sample $X$ is less than $r$ and any two elements of the sample $Y$ is less than $s,$
the strong law of large numbers for $U$-statistics  (\cite{Hoeffding}) allows us to affirm that for
any $r,s>0$, \begin{equation}\label{cs}
               RR_{n}^{X}(r)\overset{a.s.}{\rightarrow }p_{X}(r),  \ \ RR_{n}^{Y}(s)%
\overset{a.s.}{\rightarrow }p_{Y}(s) \ \ \text{and} \ \ RR_{n}^{X,Y}(r,s)\overset{a.s.}{%
\rightarrow }p_{X,Y}(r,s).
             \end{equation}
\bigskip\ We want to test $H_{0}:$ $X$ and $Y$ are independent, against $%
H_{1}:$ $H_{0}$ does not hold.

If $H_{0}$ is true, then $p_{X,Y}(r,s)=p_{X}(r)p_{Y}(s)$ for all $r,s>0$, and we expect
that if $n$ is large,  $RR_{n}^{X,Y}(r,s)\cong RR_{n}^{X}(r)RR_{n}^{Y}(s)$ for
any $r, s >0.$
In  \cite{Fernandez} it is proposed to reject $H_{0}$ when $T_n >c$, where 
\begin{equation}\label{cvm}
T_n:= n\int_0^{+ \infty} \int_0^{+ \infty}\left (
RR_{n}^{X,Y}(r,s)-RR_{n}^{X}(r)RR_{n}^{Y}(s)\right )^{2}dG(r,s),  
 \end{equation}%
where $c$ is a constant and $G$ is a properly chosen distribution function.
Observe that $T_n$ is a functional of the $L^{2}$ Cram\'er--von Mises type applied to the process 
 $\{E_n(r,s)\}_{r,s>0}$  where 
\begin{equation}\label{En}
 E_n(r,s):=\sqrt{n} \left (RR_{n}^{X,Y}(r,s)- RR_{n}^{X}(r)RR_{n}^{Y}(s) \right ).
\end{equation}

The theoretical results established in \cite{Fernandez}  about the process defined in (\ref{En}), 
are valid for any distance functions $d_X$ and $d_Y$, and remains
valid if we consider other continuous functionals such as an $L^{1}$-Cram\'er--von Mises type or
one of Kolmogorov--Smirnov type.

If $\left( X_{1},Y_{1}\right) ,\left( X_{2},Y_{2}\right) ,...,\left(
X_{n},Y_{n}\right) $ are i.i.d. in $S_{X}\times S_{Y},$ we will compare the $T_n$
statistic proposed in \cite{Fernandez}, which we will call $T_n^{(2)}$, defined by 
$$T_{n}^{(2)}:=n\int_{0}^{+\infty }\int_{0}^{+\infty
}\left( RR_{n}^{X,Y}\left( r,s\right) -RR_{n}^{X}\left( r\right)
RR_{n}^{Y}\left( s\right) \right) ^{2}dG(r,s)$$ with the statistics defined as 
\begin{equation*}
T_{n}^{\left( 1\right) }:=\sqrt{n}\int_{0}^{+\infty }\int_{0}^{+\infty
}\left\vert RR_{n}^{X,Y}\left( r,s\right) -RR_{n}^{X}\left( r\right)
RR_{n}^{Y}\left( s\right) \right\vert dG(r,s)
\end{equation*}%
and 
\begin{equation*}
T_{n}^{\left( \infty \right) }:=\sqrt{n}\sup_{r,s>0}\left\vert
RR_{n}^{X,Y}\left( r,s\right) -RR_{n}^{X}\left( r\right) RR_{n}^{Y}\left(
s\right) \right\vert .
\end{equation*}

Observe that in the general case in which $X$ and $Y$ lie in metric spaces $%
\left( S_{X},d_{X}\right) $ and $\left( S_{Y},d_{Y}\right) $, the statistics 
$T_{n}^{\left( 1\right) },T_{n}^{\left( 2\right) }$ and $T_{n}^{\left(
\infty \right) }$ depend on the distance functions $d_{X}$ and $d_{Y}.$ In
Section 4 we will compare the power under several alternative tests
based on $T_{n}^{\left( 1\right) },T_{n}^{\left( 2\right) }$ and $%
T_{n}^{\left( \infty \right) }$ for different distance functions $d_{X}$ and 
$d_{Y}.$

In the case in which $X$ and $Y$ are discrete time series, we will use the
classical $l^{1},l^{2}$ and $l^{\infty }$ distances, that is, $d_{X}\left(
x,x^{\prime }\right) =\sum_{n\geq 1}\left\vert x_{n}-x_{n}^{\prime
}\right\vert ,$ $d_{X}\left( x,x^{\prime }\right) =\sqrt{\sum_{n\geq
1}\left( x_{n}-x_{n}^{\prime }\right) ^{2}}$ and $d_{X}\left( x,x^{\prime
}\right) =\sup_{n\geq 1}\left\vert x_{n}-x_{n}^{\prime }\right\vert $ and
analogously for $d_{Y}.$ Analogously, when $X$ and $Y$ are continuous time
series, we will use the classical $L^{1},L^{2},L^{\infty }$ distances, that
is, $d_{X}\left( x,x^{\prime }\right) =\int_{-\infty }^{+\infty }\left\vert
x(t)-x^{\prime }(t)\right\vert dt,$ $d_{X}\left( x,x^{\prime }\right) =\sqrt{%
\int_{-\infty }^{+\infty }\left( x(t)-x^{\prime }(t)\right) ^{2}dt}$ and $%
d_{X}\left( x,x^{\prime }\right) =\sup_{t\in \mathbb{R}}\left\vert
x(t)-x^{\prime }(t)\right\vert .$ We will use the notation $T_{n}^{\left(
i,j\right) }$ where $i,j=1,2,\infty $ for the statistic $T_{n}^{\left(
i\right) }$ where the distance functions used are the $l^{j}$ $\left( \text{%
or }L^{j}\right) $ distance.

In all cases, as proposed in \cite{Fernandez}, we will use a
weight function $G$ such that $dG(r,s)=g_1(r)g_2(s)drds$ where $g_{1}$ and $g_{2}$  are $%
g_{1}\left( z\right) =\varphi \left( \frac{z-\mu _{X}}{\sigma _{X}}\right) $
with $\varphi $ being  the density function of an $N\left( 0,1\right) $ random
variable and $\mu _{X}=\mathbb{E}\left( d\left( X_{1},X_{2}\right) \right) ,$
$\sigma _{X}^{2}=\mathbb{V}\left( d\left( X_{1},X_{2}\right) \right) $ being 
$X_{1},X_{2}$ independent random variables with the same distribution as $X.$
Analogously, $g_{2}\left( t\right) =\varphi \left( \frac{t-\mu _{Y}}{\sigma
_{Y}}\right) .$ In practice  $\mu _{X}$ and  $\sigma _{X}$ are unknown, but
they can be estimated naturally by $\widehat{\mu }_{X}=\frac{1}{N}%
\sum_{i\neq j}d\left( X_{i},X_{j}\right) $ and $\widehat{\sigma }_{X}^{2}=%
\frac{1}{N}\sum_{i\neq j}\left( d\left( X_{i},X_{j}\right) -\widehat{\mu }%
_{X}\right) ^{2}$ where $N=n(n-1)$, and analogously with $\widehat{\mu }_{Y}$ and $\widehat{%
\sigma }_{Y}^{2}.$
\section{Computing the Statistics}
 Given $\left( X_{1},Y_{1}\right) ,\left( X_{2},Y_{2}\right) ,...,\left(
X_{n},Y_{n}\right) $  i.i.d. in $S_{X}\times S_{Y},$ and choosen the weight functions $g_1$, $g_2$
to be used, the statistic $T_n^{(1)}, T_n^{(2)}$ and $T_n^{(\infty)}$ can be computed in the steps
that indicated the following three propositions.

\begin{Proposition}\label{T2}
Calculation of $T_n^{(2)}$.\\
Step 1. Compute $d\left( X_{i},X_{j}\right) $ and $d\left(
Y_{i},Y_{j}\right) $ for all $i,j\in \left\{ 1,2,3,...,n\right\} $ where $%
i\neq j$ and put $N=n(n-1).$

Step 2. Re-order $\left\{ d\left( X_{i},X_{j}\right) \right\} _{i\neq j}$ as $%
Z_{1},Z_{2},...,Z_{N}$ such that $Z_{1}<Z_{2}<...<Z_{N}$ and $\left\{
d\left( Y_{i},Y_{j}\right) \right\} _{i\neq j}$ as $T_{1},T_{2},...,T_{N}$
maintaining the same indexing as $Z^{\prime }s$ (that is, if $d\left(
X_{i},X_{j}\right) =Z_{h}$ then $d\left( Y_{i},Y_{j}\right) =T_{h}$).

Step 3. Compute the order statistics for $T^{\prime }s,$ that is, $%
T_{1}^{\ast }<T_{2}^{\ast }<...<T_{N}^{\ast }.$

Step 4. Compute 
\begin{equation*}\label{a}
A_n=\frac{1}{N^{2}}\sum_{i=1}^{N}\sum_{j=1}^{N}\left( 1-G_1 \left( \max \left\{
Z_{i},Z_{j}\right\} \right) \right) \left( 1-G_2 \left( \max \left\{
T_{i},T_{j}\right\} \right) \right),
\end{equation*}
\begin{equation*}\label{b}
B_n=\left (1-\frac{1}{N^{2}}\sum_{i=1}^{N}\left( 2i-1\right) G_1
\left( Z_{i}\right)  \right )\left (1-\frac{1}{N^{2}}\sum_{i=1}^{N}\left( 2i-1\right) G_2
\left( T_{i}^{\ast }\right)  \right ),
\end{equation*}

\begin{equation*}\label{c}
C_n=\frac{1}{N^{3}}\sum_{i=1}^{N}\sum_{j=1}^{N}\sum_{k=1}^{N}\left( 1-G_1
\left( \max \left\{ Z_{i},Z_{j}\right\} \right) \right) \left( 1-G_2 \left(
\max \left\{ T_{i},T_{k}\right\} \right) \right) .
\end{equation*}

Step 5. Compute \begin{equation*}\label{Tn}
	T_n^{(2)}=n(A_n+B_n-2C_n).
\end{equation*}

\end{Proposition}
\begin{Proposition}\label{T1}
Calculation of $T_n^{(1)}$.\\
Step 1. Compute $d\left( X_{i},X_{j}\right) $ and $d\left(
Y_{i},Y_{j}\right) $ for all $i,j\in \left\{ 1,2,3,...,n\right\} $ where $%
i\neq j$ and put $N=n(n-1).$

Step 2. Re-order $\left\{ d\left( X_{i},X_{j}\right) \right\} _{i\neq j}$ as $%
Z_{1},Z_{2},...,Z_{N}$ such that $Z_{1}<Z_{2}<...<Z_{N}$ and $\left\{
d\left( Y_{i},Y_{j}\right) \right\} _{i\neq j}$ as $T_{1},T_{2},...,T_{N}$
maintaining the same indexing as $Z^{\prime }s$ (that is, if $d\left(
X_{i},X_{j}\right) =Z_{h}$ then $d\left( Y_{i},Y_{j}\right) =T_{h}$).

Step 3. Compute the order statistics for $T^{\prime }s,$ that is, $%
T_{1}^{\ast }<T_{2}^{\ast }<...<T_{N}^{\ast }.$

Step 4. For each $h,j\in \left\{ 1,2,3,...,N-1\right\} $ compute $%
c(h,j)=\sum_{i=1}^{h}\mathbf{1}_{\left\{ T_{i}<T_{j+1}^{\ast }\right\} }$,
that is, the number of elements of the vector $\left(
T_{1},T_{2},...,T_{h}\right) $ that are less than $T_{j+1}^{\ast }$ for $%
h,j=1,2,3,...,N-1.$

Step 5. Compute $$T_{n}^{\left( 1\right) }=\frac{\sqrt{n}}{N}%
\sum_{h,j=1}^{N-1}\left( G_{1}\left( Z_{h+1}\right) -G_{1}\left(
Z_{h}\right) \right) \left( G_{2}\left( T_{j+1}^{\ast }\right) -G_{2}\left(
T_{j}^{\ast }\right) \right) \left\vert c(h,j)-\frac{jh}{N}\right\vert .$$
\end{Proposition}
\begin{Proposition}\label{Tinf}
 Calculation of $T_n^{(\infty)}$.\\
 Step 1. Compute $d\left( X_{i},X_{j}\right) $ and $d\left(
Y_{i},Y_{j}\right) $ for all $i,j\in \left\{ 1,2,3,...,n\right\} $ where $%
i\neq j$ and put $N=n(n-1).$

Step 2. Re-order $\left\{ d\left( X_{i},X_{j}\right) \right\} _{i\neq j}$ as $%
Z_{1},Z_{2},...,Z_{N}$ such that $Z_{1}<Z_{2}<...<Z_{N}$ and $\left\{
d\left( Y_{i},Y_{j}\right) \right\} _{i\neq j}$ as $T_{1},T_{2},...,T_{N}$
maintaining the same indexing as $Z^{\prime }s$ (that is, if $d\left(
X_{i},X_{j}\right) =Z_{h}$ then $d\left( Y_{i},Y_{j}\right) =T_{h}$).

Step 3. Compute the order statistics for $T^{\prime }s,$
that is,  $T_{1}^{\ast
}<T_{2}^{\ast }<...<T_{N}^{\ast }.$

Step 4. Compute the $\left( N-1\right) \times \left( N-1\right) $ matrix $C$
such that $$C_{ij}=\left\vert \sum_{k=1}^{N}\mathbf{1}_{\left\{ Z_{k}\leq
Z_{i},\text{ }T_{k}\leq T_{j}^{\ast }\right\} }-\frac{ij}{N}%
\right\vert .$$

Step 5. Compute $$T_{n}^{\left( \infty \right) }=\frac{%
\sqrt{n}}{N}\max_{i,j}C_{ij}.$$
\end{Proposition}

\section{Simulation Study}
When $X$ and $Y$ lie in high dimensional spaces, it is interesting to analyse the performance of the test 
statistics $T_n^{(1)},T_n^{(2)}$ and 
$T_n^{(\infty)}$ for different distance functions $d_X$ and $d_Y$.
In this section we will compare the power of the $9$ test statistics $T_n^{(i,j)}$ for $i,j=1,2,\infty$ in 
the cases in which $X$ and $Y$ are
discrete and continuous time series under several alternatives. In all cases we will use the same distance function for $X$ and $Y$, that is, if $X$ and
$Y$ are discrete time series, then we will use $l^{j}$ for both $X$ and $Y$ for $j=1,2,\infty$, and 
analogously in the case in which $X$ and $Y$
are continuous time series. In all cases, $X$ and $Y$ are time series of length $100$ and the power (due to the 
computational cost) was calculated at the
 $5 \%$ level from $500$ replications. Every $p$-value was calculated by a
permutation method (as suggested in  \cite{Fernandez}) for $100$ replications. 
\subsection{The discrete case}
We analyse two scenarios for $X$: one of them is when $X$ is  AR$(1)$ where $\phi=0.1$, which 
we call simply AR$(0.1)$ and the second case, is when $X$ is ARMA$(2,1)$ with parameters
$\phi=(0.2,0.5)$ and $\theta=0.2$. In both cases we consider three possible $Y$: $Y_{1}=X^{2}+3\varepsilon$,
$Y_2=\sqrt{|X|}+\sigma\varepsilon$  where $\sigma^{2}$  means
the variance of $\sqrt{|X|}$ and
$Y_3=\varepsilon X$. In all cases, $\varepsilon$ is a standard Gaussian white noise $(N(0,1))$ independent of
$X$. In Table \ref{arn30} and Table \ref{arn50} we show the power for $n=30$ and
$n=50$, respectively, in the case in which $X$ is an AR$(0.1)$
process for the 9 tests considered. Similarly  Table \ref{arman30} and Table \ref{arman50} show the power for
the case in which $X$ is an ARMA$(2,1)$ process. Tables 1--4 do 
not show important differences between using $T_n^{(2)}$, $T_n^{(1)}$ or $T_n^{(\infty)}$.  In Figure
\ref{powerarma} we show  the power as a function of sample size, where the statistic considered is $T_n^{(2)}$, that is $T_n^{(2,1)}$, $T_n^{(2,2)}$
and $T_n^{(2,\infty)}$. The behaviour of $T_n^{(1)}$ and $T_n^{(2)}$ is similar. Figure \ref{powerarma}
suggest that the power increases as the distance
function considered goes from $d_\infty$ ($l^{\infty}$ distance) to $d_1$ ($l^{1} $ distance). Also for the
alternative $Y=Y_2$, the statistic 
based on the $l^{\infty}$ distance has difficulties in detecting the dependence between $X$ and $Y$ 
(which grows very slowly as $n$ increases), while for $n=60$ the power of the test based on $l^{1}$ or $l^{2}$
distances is near unity.

\begin{table*}[ht]
\caption {Comparison of powers,  at the $5 \%$ level, for the different tests, where  $X$ is AR$(0.1)$ and
$Y_{1}=X^{2}+3\varepsilon$, $Y_2=\sqrt{|X|}+\sigma
\varepsilon$, $Y_3=\varepsilon X$  for sample size of $n=30$.}
\label{arn30}

\centering
\begin{small}\begin{tabular}{|c|ccc|ccc|ccc|}

    \hline
  $n=30$ & $T_{n}^{(1,1)}$  & $T_{n}^{(1,2)}$  & $T_{n}^{(1,\infty)}$  & $T_{n}^{(2,1)}$  & $T_{n}^{(2,2)}$  & $T_{n}^{(2,\infty)}$  &$T_{n}^{(\infty,1)}$  &  $T_{n}^{(\infty, 2)}$ & $T_{n}^{(\infty,\infty)}$ \\
   \hline 
   $Y=Y_1$ &0.39 & \textbf{0.40} & 0.29 &0.39 &\textbf{0.40} &0.24 & 0.32 & 0.28 & 0.16 \\
   \hline
   $Y=Y_2$ &0.45 & 0.22 & 0.10 &\textbf{0.71} &0.52 &0.11 &0.69 &0.19 &0.05  \\
   \hline
   $Y=Y_3$ &0.91 & 0.79 & 0.28 & 0.87 & 0.77 & 0.34  & \textbf{0.92} & 0.77 & 0.27 \\
  
   \hline
   
\end{tabular}             \end{small}
\end{table*}

\begin{table*}[ht]

\caption {Comparison of powers,  at the $5 \%$ level, for the different tests, where  $X$ is AR$(0.1)$ and
$Y_{1}=X^{2}+3\varepsilon$, $Y_2=\sqrt{|X|}+\sigma
\varepsilon$, $Y_3=\varepsilon X$  for sample size of $n=50$.}
\label{arn50}

\centering
\begin{small}\begin{tabular}{|c|ccc|ccc|ccc|}
		
		\hline
		$n=50$ & $T_{n}^{(1,1)}$  & $T_{n}^{(1,2)}$  & $T_{n}^{(1,\infty)}$  & $T_{n}^{(2,1)}$  & $T_{n}^{(2,2)}$  & $T_{n}^{(2,\infty)}$  &$T_{n}^{(\infty,1)}$  &  $T_{n}^{(\infty, 2)}$ & $T_{n}^{(\infty,\infty)}$ \\
		\hline 
		$Y=Y_1$ & 0.90 & \textbf{0.92} &0.74 & 0.54 &0.59 &0.47 &0.49 & 0.57  & 0.34 \\
		\hline
		$Y=Y_2$ &0.34 &0.50  &0.38 & \textbf{0.97} &0.81 &0.16 &0.92 & 0.89 &  0.79 \\
		\hline
		$Y=Y_3$ &\textbf{1.00 }&0.94  &0.64 & \textbf{1.00} &0.94 &0.61 &0.99 & 0.92  & 0.57 \\
		
		\hline
		
\end{tabular}             \end{small}
\end{table*}

\begin{table*}[ht]
\caption {Comparison of powers,  at the $5 \%$ level, for the different tests, where  $X$ is ARMA$(2,1)$ and
$Y_{1}=X^{2}+3\varepsilon$, $Y_2=\sqrt{|X|}+\sigma
\varepsilon$, $Y_3=\varepsilon X$  for sample size of $n=30$.}\label{arman30}

\centering
\begin{small}\begin{tabular}{|c|ccc|ccc|ccc|}

    \hline
   $n=30$ & $T_{n}^{(1,1)}$  & $T_{n}^{(1,2)}$  & $T_{n}^{(1,\infty)}$  & $T_{n}^{(2,1)}$  & $T_{n}^{(2,2)}$  & $T_{n}^{(2,\infty)}$  &$T_{n}^{(\infty,1)}$  &  $T_{n}^{(\infty, 2)}$ & $T_{n}^{(\infty,\infty)}$ \\
   \hline 
   $Y=Y_1$ &\textbf{0.85} & 0.82 & 0.53 &0.83 &0.78 &0.57 & 0.77 & \textbf{0.85} & 0.47 \\
   \hline
   $Y=Y_2$ &0.56 & 0.27 & 0.08 &\textbf{0.84} &0.78 &0.34 &0.49 &0.34 &0.08  \\
   \hline
   $Y=Y_3$ &\textbf{1.00 }& \textbf{1.00} & 0.93 & 0.99 & 0.93 & 0.50  & 0.96 & 0.88 & 0.23 \\
  
   \hline
   
\end{tabular}             \end{small}
\end{table*}

\begin{table*}[ht]

\caption {Comparison of powers,  at the $5 \%$ level, for the different tests, where  $X$ is ARMA$(2,1)$ and
$Y_{1}=X^{2}+3\varepsilon$, $Y_2=\sqrt{|X|}+\sigma
\varepsilon$, $Y_3=\varepsilon X$  for sample size of  $n=50$.}\label{arman50}

\centering
\begin{small}\begin{tabular}{|c|ccc|ccc|ccc|}
		
		\hline
		$n=50$ & $T_{n}^{(1,1)}$  & $T_{n}^{(1,2)}$  & $T_{n}^{(1,\infty)}$  & $T_{n}^{(2,1)}$  & $T_{n}^{(2,2)}$  & $T_{n}^{(2,\infty)}$  &$T_{n}^{(\infty,1)}$  &  $T_{n}^{(\infty, 2)}$ & $T_{n}^{(\infty,\infty)}$ \\
		\hline 
		$Y=Y_1$ &\textbf{0.98} &\textbf{0.98}  & 0.82 &0.96 &0.97 &0.82 & 0.96 &0.96 & 0.82 \\
		\hline
		$Y=Y_2$ &0.76 &0.48  & 0.04 &\textbf{0.98} &\textbf{0.98} &0.43 & 0.65 &0.43 & 0.03 \\
		\hline
		$Y=Y_3$ &\textbf{1.00} &\textbf{1.00}  &0.74  &\textbf{1.00} &\textbf{1.00} &0.77 & \textbf{1.00} &\textbf{1.00} & 0.75 \\
		
		\hline
		
\end{tabular}             \end{small}
\end{table*}

\begin{figure}[h!]
\begin{center}
\begin{tabular}{ll}
\includegraphics[scale=0.7]{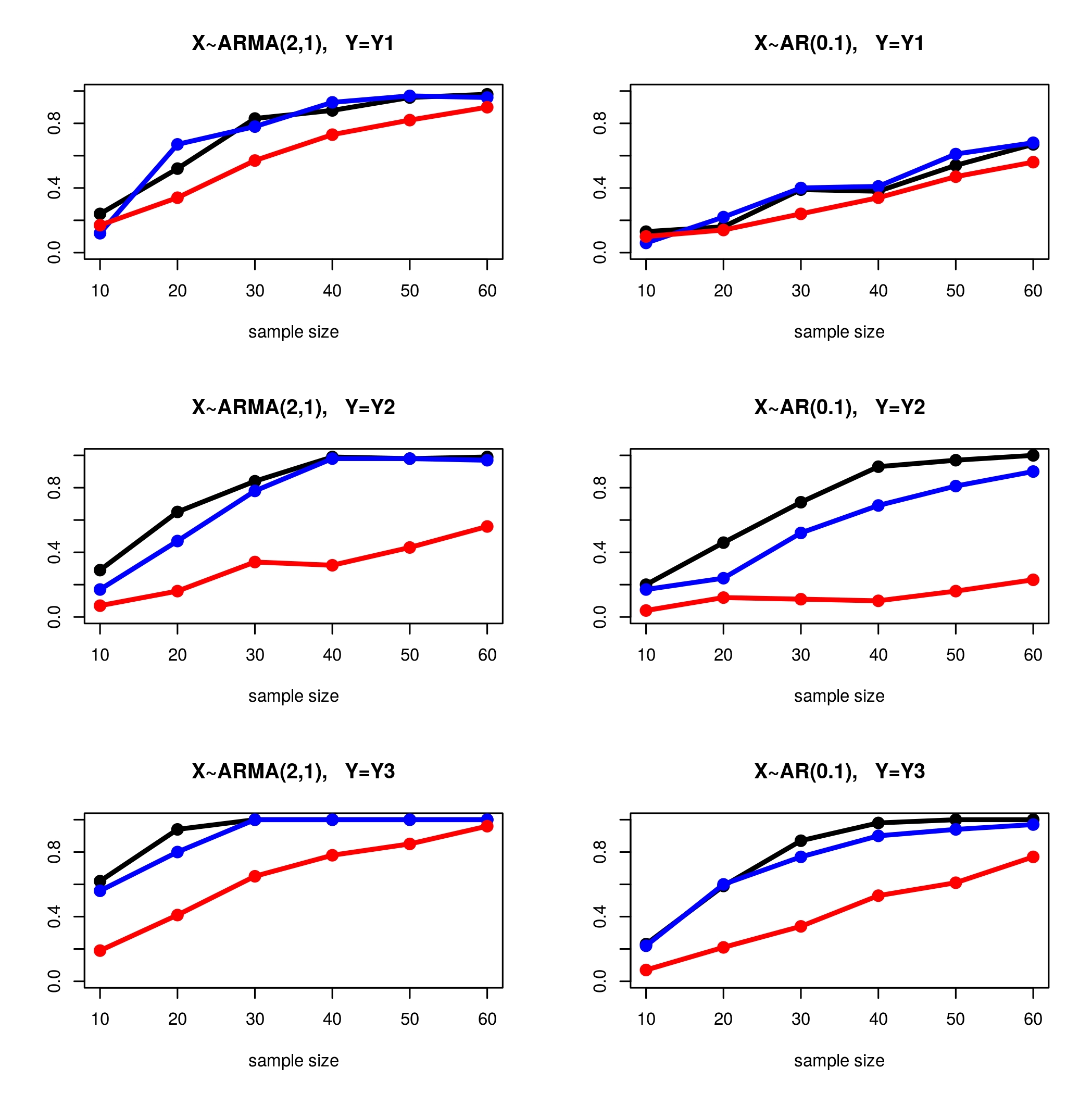}

\end{tabular}
\end{center}
\caption{Power at  $5 \%$ level under several alternatives for the statistic $T_n^{(2)}$ using Manhattan distance ($T_n^{(2,1)}$ in black),
Euclidean distance ($T_n^{(2,2)}$ in blue) and maximum distance ($T_n^{(2,\infty)}$ in red). $Y_1$, $Y_2$ 
and $Y_3$ are defined  in Tables 1--4.}\label{powerarma}
\end{figure}

\subsection{The continuous case}
In this subsection, we will take $X$ to be a fractional Brownian motion with $\sigma =1$ observed in $[0,1]$ (at times $0,1/100,2/100,...,99/100$)
for $H=0.5$ (standard Brownian motion) and $H=0.7$. We consider $7$ dependence cases  between $X$ and $Y$. The first three are for the case in which
$X$ is a standard Brownian motion ($Bm$) and the dependence is defined by  $Y_{1}=X^{2}+3\varepsilon$, $Y_2=\sqrt{|X|}+
\varepsilon$, $Y_3=\varepsilon X+3\varepsilon'$ where $\varepsilon$ and $\varepsilon'$ are Gaussian white 
noises with $\sigma=1$ such that $X,\varepsilon$
and $\varepsilon'$ are independent. In the last $4$ alternatives, we explore the power when $Y$ is a linear 
functional of $X$.
More explicitly, we will consider the case in which $Y$ is a fractional Ornstein--Uhlenbeck process driven by a 
Brownian motion ($X$) for $H=0.5$ ($Bm$) and fractional Brownian motion for $H=0.7$ ($fBm$), 
which we call  the $OU$ and $FOU$ processes, respectively. 
A particular, a linear combination of $FOU$, which we call $FOU(2)$, and whose definition, theoretical 
development and simulations are
found in \cite{Kalemkerian} and \cite{Kalemkerian_discreto},  is a particular case
of the models proposed in \cite{Arratia}.  More explicitly,  the $FOU$ process is defined by 
$Y_t=\sigma \int_{-\infty}^{t}e^{-\lambda (t-s)}dX_s$ 
(where $X=\{X_{t}\}$ is an $fBm$), and the $FOU(2)$ process is defined by 
$Y_t= \dfrac{\lambda_{1}}{\lambda_{1}-\lambda_{2}}\sigma \int_{-\infty}^{t}e^{-\lambda_{1} (t-s)}dX_s+
\dfrac{\lambda_{2}}{\lambda_{2}-\lambda_{1}}\sigma \int_{-\infty}^{t}
e^{-\lambda_{2} (t-s)}dX_s $  (where $X=\{X_{t}\}$ is an $fBm$). When $H=0.5$,
we will call them simply the  $OU$ and $OU(2)$ processes, as defined in \cite{Arratia}.
Tables \ref{mBfn30} and \ref{mBfn50} gives us the power for $n=30$ and $n=50$, respectively, for the $7$
alternatives. In these tables, $Y_4$ means an
$ OU$ process driven by $X$ with parameters $\sigma=1, \lambda=0.3$. Similarly $Y_5$ is an $ FOU$ 
process with  parameters $\sigma=1, H=0.7, \lambda=0.3$, 
$Y_6 \sim OU(2)$ with parameters $\sigma=1, \lambda_{1}=0.3, \lambda_{2}=0.8$ and $Y_7 \sim FOU(2)$ with parameters $\sigma=1,
H=0.7, \lambda_{1}=0.3, \lambda_{2}=0.8$. In this paper we do not present the performance of the test 
for other choices of the parameters, because the 
behaviour is similar. As expected, for values of $\sigma$ larger than $1$, the dependence between $X$ and $Y$ is more difficult to detect, and we need
to increase the sample size. The same occurs if we take $\lambda_1$ near to $\lambda_2$ in $OU(2)$ and $FOU(2)$.
Tables \ref{mBfn30} and \ref{mBfn50} show, as in the discrete case,  no substantial differences
between  the performance of the three statistics ($T_n^{(1)}, T_n^{(2)}$ or $T_n^{(\infty)}$). With respect 
to which distance between elements of $X$ and $Y$ is more appropriate, Table  \ref{mBfn30} and Table 
\ref{mBfn50} show that
the $L^{\infty}$ distance performs poorly under the alternatives $Y_1, Y_2$ and $Y_3$, but better 
 under alternatives 
$Y_4, Y_5, Y_6 $ and $Y_7$. The performance of the $L^{1}$ and $L^{2}$ distances is similar throughout the 
$7$ alternatives. Figure 
\ref{powerT1T2Tinf} expands the information given in Table \ref{mBfn30} and Table \ref{mBfn50} for the cases 
$Y_1,Y_2$ and $Y_3$ because it shows us
 the power for the statistics $T_n^{(i,j)}$ for $i,j=1,2,\infty$ for sample sizes of $n=10$ to $n=50$. Figure 
\ref{powerT1T2Tinf} show clearly that the $L^{\infty}$ distance has a 
poorer performance  than the $L^{1}$ and $L^{2}$ distances. 
Figure \ref{powerfou} shows the power as a function of sample size for the statistic $T_n^{(2)}$ in the cases $Y_4, Y_5, Y_6$ and $Y_7$. The 
behaviour of the statistics $T_n^{(1)}$ and  $T_n^{(2)}$ is similar. Contrary to what happened in the cases
$Y_1,Y_2$ and $Y_3$, the 
$L^{\infty}$ distance performs clearly better than the $L^{1}$ and $L^{2}$ distances. Also, 
Figure \ref{powerfou} show that the performance
of the $T_n^{(2)}$ statistic increases as we move from the use of the $L^{1}$ distance to the use of the 
$L^{\infty}$ distance. On the other hand, Figure \ref{powerfou} shows that
the power in the case of the $OU$ alternative is higher than for the $OU(2)$ alternative (and the same 
for $FOU$ versus $FOU(2)$), which is reasonable, because
the dependence between $X$ and $Y$ is  simpler in the $OU$ ($FOU$) case than in the $OU(2)$ ($FOU(2)$) case.
Also, the power in the 
$OU$ ($OU(2)$) case is higher than in the  $FOU$ ($FOU(2)$) case, 
which is to be  expected because when $H=0.7$, the fractional Brownian motion has a long 
range dependence, therefore it is reasonable that the dependence between $X$ and $Y$ is more difficult to detect.\\
To conclude this section, observe that the test of independence based on recurrence rates has a power that 
grows as $n$ grows for the $9$ statistics
considered, $T_n^{(i,j)}$ for $i,j=1,2,\infty$ (as expected according to the theory developed in \cite{Fernandez}) in all the alternatives considered 
for both the discrete and the continuous cases. In most of the cases, the test has a power near to
unity for  moderately small sample sizes. Taking into account
what was observed in this section, it can be said that there is no preference to use the test based 
on $T_n^{(1)}, T_n^{(2)}$ or $T_n^{(\infty)}$, but 
in the three cases, in general the performance is better as the function distance goes from the
$L^{1}$ ($l^{1}$) distance to the $L^{\infty}$ ($l^{\infty}$) 
distance in 
some cases, and in the opposite direction for other cases. Therefore,  it can be suggested that one
 use the  test
statistic using  the $L^{1}$ ($l^{1}$) or $L^{2}$ ($l^{2}$) distance
and the $L^{\infty}$ ($l^{\infty}$) distance to cover both possibilities.

\begin{table*}[ht]
\caption {Comparison of powers,  at the $5 \%$ level, of the different tests, where $X\sim Bm$ in alternatives 
$Y_1,Y_2,Y_3,Y_4,Y_6$ and $X\sim fBm$ with $H=0.7$ in alternatives $Y_5, Y_7$ where $Y_{1}=X^{2}+3\varepsilon$, $Y_2=\sqrt{|X|}+
\varepsilon$, $Y_3=\varepsilon X+3\varepsilon'$, $Y_4=OU, Y_5=FOU, Y_6=OU(2),$ and $ Y_7=FOU(2)$  for 
sample size of $n=30$.}
\label{mBfn30}

\centering
\begin{small}\begin{tabular}{|c|ccc|ccc|ccc|}

    \hline
   $n=30$ & $T_{n}^{(1,1)}$  & $T_{n}^{(1,2)}$  & $T_{n}^{(1,\infty)}$  & $T_{n}^{(2,1)}$  & $T_{n}^{(2,2)}$  & $T_{n}^{(2,\infty)}$  &$T_{n}^{(\infty,1)}$  &  $T_{n}^{(\infty, 2)}$ & $T_{n}^{(\infty,\infty)}$ \\
   \hline 
   $Y=Y_1$ &0.70 & 0.58 & 0.38 &0.79 &0.76 &0.57 & 0.66 & \textbf{0.83} & 0.47 \\
   \hline
   $Y=Y_2$ &0.44 & 0.43 & 0.27 &0.51 &0.52 &0.22 &0.51 &\textbf{0.54} &0.22  \\
   \hline
   $Y=Y_3$ &0.33 & \textbf{0.42}& 0.29 & \textbf{0.42} & 0.41 & 0.19  & 0.39 & 0.37 & 0.20 \\ 
     \hline
     $Y=Y_4$ &0.69 &0.79 &\textbf{0.92} &0.58 &0.67 &0.96 &0.43 &0.56 &0.54 \\
     \hline
     $Y=Y_5$ &0.30 &0.37 &0.53 &0.30 &0.46 &\textbf{0.81} & 0.54 &0.56 &0.25\\
     \hline
     $Y=Y_6$ &0.16 &0.21 &0.15 &0.31 &0.38 &\textbf{0.74} &0.17 &0.16 &0.18 \\
     \hline
     $Y=Y_7$ &0.07 &0.15 &\textbf{0.95} &0.18 &0.21 &0.43 &0.07 &0.11 &0.02 \\
     \hline

\end{tabular}             \end{small}
\end{table*}

\begin{table*}[ht]

\caption {Comparison of powers,  at the $5 \%$ level, of the different tests, where $X\sim Bm$ in alternatives 
$Y_1,Y_2,Y_3,Y_4,Y_6$ and $X\sim fBm$ with $H=0.7$ in alternatives $Y_5, Y_7$ where $Y_{1}=X^{2}+3\varepsilon$, $Y_2=\sqrt{|X|}+
\varepsilon$, $Y_3=\varepsilon X+3\varepsilon'$, $Y_4=OU, Y_5=FOU, Y_6=OU(2),$ and $ Y_7=FOU(2)$  for 
sample size of $n=50$.}
\label{mBfn50}

\centering
\begin{small}\begin{tabular}{|c|ccc|ccc|ccc|}
		
		\hline
		$n=50$ & $T_{n}^{(1,1)}$  & $T_{n}^{(1,2)}$  & $T_{n}^{(1,\infty)}$  & $T_{n}^{(2,1)}$  & $T_{n}^{(2,2)}$  & $T_{n}^{(2,\infty)}$  &$T_{n}^{(\infty,1)}$  &  $T_{n}^{(\infty, 2)}$ & $T_{n}^{(\infty,\infty)}$ \\
		\hline 
		$Y=Y_1$ &0.90 &0.91  & 0.90 &0.93 &\textbf{0.95} &0.84 & 0.89 &0.92 & 0.79 \\
		\hline
		$Y=Y_2$ &0.70 &0.78  & 0.46 &0.82 &0.80 &0.44 & 0.63 &\textbf{0.85} & 0.33 \\
		\hline
		$Y=Y_3$ &\textbf{0.68} & 0.67  & 0.41 &0.51 &0.62 &0.44 & 0.61 &0.67 & 0.38 \\
		
		\hline
		 \hline
     $Y=Y_4$ &\textbf{1.00} &0.86 &\textbf{1.00} &0.75 &0.92 &0.99 &0.74 &0.70 &0.31 \\
     \hline
     $Y=Y_5$ &0.33 &0.46 &0.97 &0.36 &0.58 &\textbf{0.98} & 0.70 &0.64 &0.29\\
     \hline
     $Y=Y_6$ &0.40 &0.48 &0.46 &0.39 &0.58 &\textbf{0.95} &0.46 &0.55 &0.78 \\
     \hline
     $Y=Y_7$ &0.06 &0.30 &\textbf{1.00} &0.22 &0.33 &0.72 &0.17 &0.21 &0.32 \\
     \hline
\end{tabular}             \end{small}
\end{table*}

\begin{figure}[h!]
\begin{center}
\begin{tabular}{ll}
\includegraphics[scale=0.8]{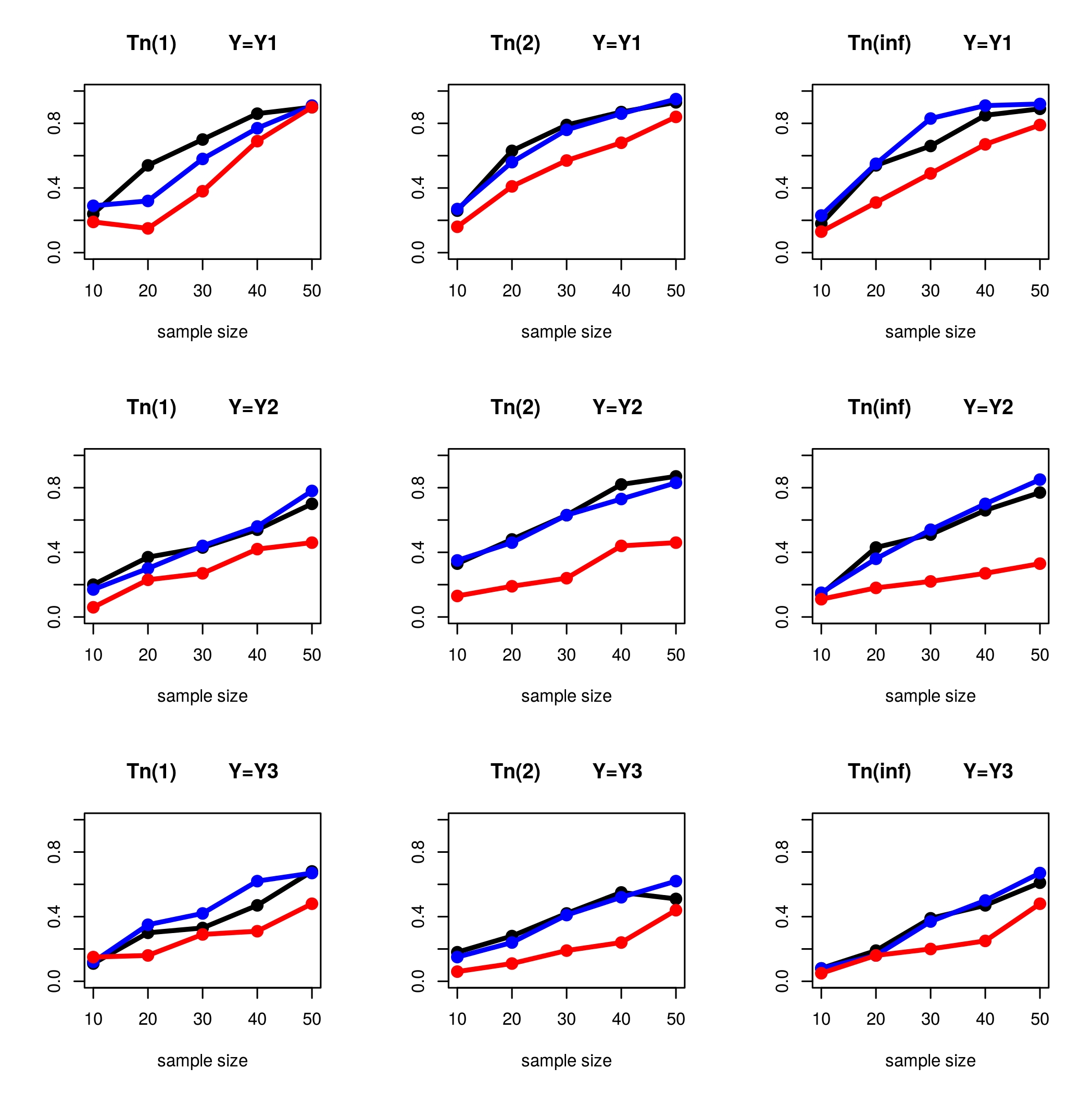}

\end{tabular}
\end{center}
\caption{Comparison of powers,  at the $5 \%$ level,  where $X$ is a standard Brownian motion,
under several alternatives for the statistics $T_n^{(1)},
T_n^{(2)}$ and $T_n^{(\infty)}$ using the Manhattan distance ($T_n^{(i,1)}$ in black),
Euclidean distance ($T_n^{(i,2)}$ in blue) and maximum ($T_n^{(i,\infty)}$ in red) for $i=1,2,\infty$. 
$Y_1, Y_2$ and $Y_3$ are defined as in Tables \ref{mBfn30} and \ref{mBfn50}.}\label{powerT1T2Tinf}
\end{figure}

\begin{figure}[h!]
\begin{center}
\begin{tabular}{ll}
\includegraphics[scale=0.7]{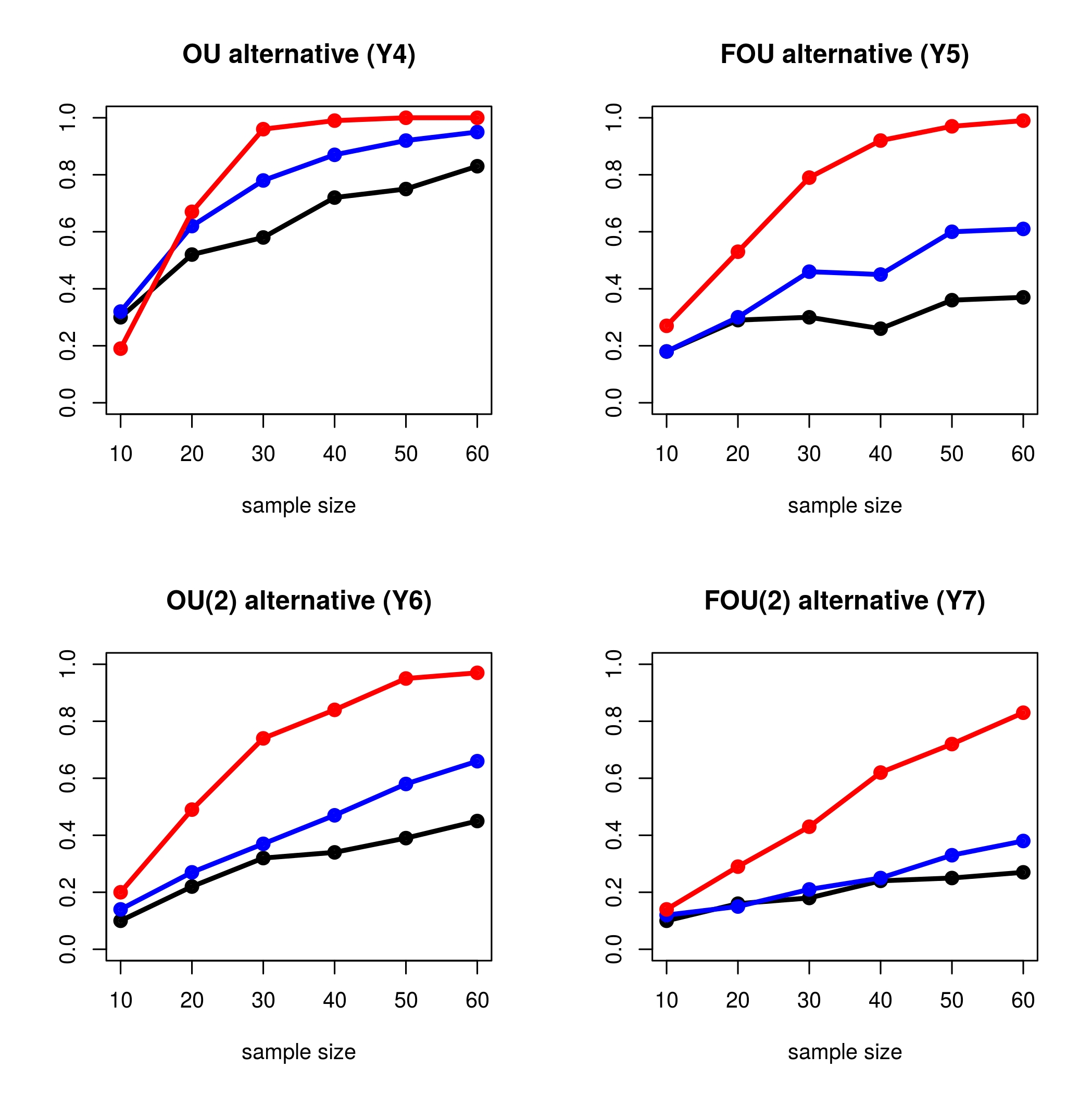}

\end{tabular}
\end{center}
\caption{Power at the $5 \%$ level, under several alternatives for the statistic $T_n^{(2)}$ using
the Manhattan distance ($T_n^{(2,1)}$ in black),
Euclidean distance ($T_n^{(2,2)}$ in blue) and maximum ($T_n^{(2,\infty)}$ in red).}\label{powerfou}
\end{figure}

\section{Comparison with Other Tests in High Dimension}
In \cite{Fernandez}  the very good performance of the recurrence rates test for random variables 
and random vectors was shown.
In this section, we will compare our test when $X$ and $Y$ lie in high dimensional spaces.
According to what was shown in the previous section, we have considered the test using the $T_n^{(2,2)}$
and $T_n^{(2,\infty)}$ statistics.
We will consider three competitors: the well known distance covariance test
proposed in \cite{Szekely} and adapted to perform better in high dimensions in \cite{Szekely2013},
the Hilbert--Schmidt Information Criterion proposed in \cite{dhsic}, and that  proposed more recently 
in 
\cite{random_projections} based on random projections. Basically, this test is based on the idea
of choosing $K$ pairs of random directions, and observing that if $X$ and $Y$ are independent, then the 
projections of $X$ and $Y$ in each one of $K$ pairs of directions are independent. This test is universally
consistent.
To perform this test, it is necessary to previously choose the number of pairs of projections ($K$), and then
$K$ independence hypothesis tests are performed. If at least one of these tests rejects the hypotheses of
independence, then $H_0$ is rejected. To work at the $5 \%$level, in \cite{random_projections}
it is proposed to use a Bonferroni correction, that is, to compute the proportion of $p$-values smaller than
$0.05/K$ to perform each one of the $K$ uni-dimensional tests. We will call the RPK test.
In Table \ref{t7} we present a power comparison at the $5\%$ level,  when $X$ is a realization of a discrete time 
series of length $100$ in three possible scenarios, where there are three alternatives for $Y$ in each scenario.
The performance of RPK is very bad in these cases, and the power using the Bonferroni correction is $0$.
For this reason we present in Table \ref{t7} the power of the RPK test using $0.05/4$ instead of $0.05/K$
for $K=100$ random projections.
Table \ref{t7} shows that our test based on $T_n^{(2,2)}$ outperforms the other tests in the $9$ cases 
considered.
Table \ref{t8} shows a comparison at the $5\%$ level, of the powers in $12$ scenarios in which $X$ and $Y$
are realizations of a continuous time series viewed at $100$ equispaced points in $[0,1].$
In this table, we have considered the RPK test for $K=5$ random projections and we have used the Bonferroni 
correction. We  chose $K=5$ projections because this is the value of $K$ for which the power of the RPK
test reaches its maximum.
Table \ref{t8} shows that our test based on $T_n^{(2,2)}$ or $T_n^{(2,\infty)}$ outperforms the other
competitors in $6$ scenarios, and the HSIC, DCOV and RPK tests have the best
performance in two cases each.

\begin{table*}[h!]\caption{ Comparison at the $5\%$ level of the powers of the $4$
independence tests considered in the case of discrete time series and different sample sizes. The parameters
in the case ARMA$(2,1)$  are $\phi=(0.2,0.5)$ and $\theta =0.2$.   The parameter $\sigma$ in 
$Y=\sqrt{|X|}+\sigma \varepsilon$ denotes the standard deviation of $\sqrt{|X|}$. $\varepsilon$ 
denotes a white noise  
with $\sigma=1$, independent of $X$.}
 \label{t7}

\centering
\begin{tabular}{c|c|ccccc}

    \hline
     $X \sim $ ARMA(2,1) & $n$& RPK  &  HSIC  &   DCOV   & $T_n^{(2,2)}$ & $T_n^{(2,\infty)}$  \\
    \hline 
     $Y=X^{2} +3\varepsilon$ & $30$ &0.533 &0.324 &0.282 &\textbf{0.785} &0.566\\ 
      & $50$ &0.570 &0.381 &0.313 &\textbf{0.975} &0.825 \\
      & $100$ &0.660 &0.537 & 0.377& \textbf{1.000} &0.994 \\
  \hline
     $Y=\sqrt{|X|}+\sigma \varepsilon $  & $30$ &0.549 &0.294 &0.240 & \textbf{0.779} &0.341\\ 
      & $50$ &0.592 &0.364 &0.270 & \textbf{0.976} & 0.427\\
      & $100$ &0.702 &0.572 &0.373 & \textbf{0.921} & 0.860\\
 \hline 
     $Y=\varepsilon X $  & $30$ &0.466 &0.501 &0.467 & \textbf{0.925} &0.498 \\ 
      & $50$ &0.468 &0.583 &0.535 & \textbf{0.996} &0.778\\
      & $100$ &0.473 &0.674 &0.567 & \textbf{1.000} &0.984\\
  \hline
    \hline
     $X \sim $ AR(0.1) & $n$& RPK &  HSIC   &   DCOV  & $T_n^{(2,2)}$ & $T_n^{(2,\infty)}$   \\
    \hline 
     $Y=X^{2} +3\varepsilon$ & $30$ &\textbf{0.484} &0.134 &0.114 & 0.398 &0.236 \\ 
      & $50$ &0.487 &0.132 &0.134 & \textbf{0.592}  &0.465 \\
      & $100$ &0.950 &0.162 &0.131 & \textbf{0.999} & 0.834 \\
  \hline
     $Y=\sqrt{|X|}+\sigma \varepsilon $  & $30$ &0.518 &0.207 &0.182 &\textbf{0.523}&0.114  \\ 
      & $50$ &0.508 &0.222 &0.184 & \textbf{0.810} & 0.157\\
      & $100$ &0.509 &0.273 &0.217 & \textbf{0.698} &0.504 \\
 \hline 
     $Y=\varepsilon X $  & $30$ &0.474 &0.349 &0.331 & \textbf{0.772} &0.345 \\ 
      & $50$ &0.486 &0.380 &0.354 & \textbf{0.945} &0.613\\
      & $100$ &0.507 &0.404 &0.372 & \textbf{1.000} &0.938  \\
  \hline
      \hline
     $X \sim $ AR(0.9) & $n$& RPK  &  HSIC  &   DCOV   & $T_n^{(2,2)}$ & $T_n^{(2,\infty)}$   \\
    \hline 
     $Y=X^{2} +3\varepsilon$ & $30$ &0.886 &0.997 & 0.949&\textbf{1.000} &0.992 \\ 
      & $50$ &0.978 &\textbf{1.000} &0.993 & \textbf{1.000} & \textbf{1.000} \\
      & $100$ &\textbf{1.000} &\textbf{1.000} &\textbf{1.000} &\textbf{1.000} &\textbf{1.000} \\
  \hline
     $Y=\sqrt{|X|}+\sigma \varepsilon $  & $30$ &0.785 &0.887 &0.649 & \textbf{0.903} &0.778 \\ 
      & $50$ &0.924 &0.992 &0.838 & \textbf{0.998} &0.978\\
      & $100$ &\textbf{1.000} &\textbf{1.000} &0.993 & \textbf{1.000} &\textbf{1.000} \\
 \hline 
     $Y=\varepsilon X $  & $30$ &0.560 &0.933 &0.870 & \textbf{1.000} &0.948  \\ 
      & $50$ &0.562 &0.983 &0.945 & \textbf{1.000} & \textbf{1.000}\\
      & $100$ &0.570 &\textbf{1.000} &0.985 & \textbf{1.000} &\textbf{1.000}\\
  \hline

\end{tabular}
\end{table*} 

\begin{table*}[p]\caption{ Comparison at the $5\%$ level of the powers of the $4$ 
		independence tests considered in the case of continuous time series and different sample sizes. 
		Bm and fBm denote a Brownian motion and fractional Brownian motion with $H=0.7$.
		$\varepsilon$ and $\varepsilon'$ are independent white noises  
		with $\sigma=1$ (and independent of $X$).}
	\label{t8}
	
	\centering
	\begin{tabular}{c|c|ccccc}
		
		\hline
		$X \sim $ Bm & $n$& RPK  &  HSIC  &   DCOV   & $T_n^{(2,2)}$ & $T_n^{(2,\infty)}$  \\
		\hline 
		$Y=X^{2} +3\varepsilon$ & $30$ &0.767 &\textbf{0.815} &0.596 &0.757 & 0.570 \\ 
		& $50$ &0.951 &\textbf{0.977} &0.829 &0.947 &0.836 \\
		& $80$ &0.999 &\textbf{1.000} &0.977 &0.994 &0.968 \\
		\hline
		$Y=\sqrt{|X|}+\varepsilon $  & $30$ &\textbf{0.954} &0.906 &0.550 &0.519 &0.224 \\ 
		& $50$ &\textbf{0.998} &0.996 &0.862 &0.802 &0.416\\
		& $80$ &\textbf{1.000} &\textbf{1.000} &0.992 &0.923 &0.834\\
		\hline 
		$Y=\varepsilon X+3 \varepsilon ' $  & $30$ &0.054 &0.099 &0.118 &\textbf{0.421} &0.185  \\ 
		& $50$ &0.050 &0.119 &0.125 &\textbf{0.619} &0.437 \\
		& $80$ &0.056 &0.128 &0.126 &\textbf{0.839 }& 0.648 \\
		\hline
		\hline
		$Y \sim $ OU  & $30$ &0.765 &0.770 &0.826 &0.651 & \textbf{0.956}\\ 
		& $50$ &0.927 &0.977 &0.988 &0.906 &\textbf{1.000} \\
		& $80$ &0.992 &\textbf{1.000} &\textbf{1.000 }&0.986 &\textbf{1.000} \\
		\hline 
		$Y \sim$ OU(2)  & $30$ &0.574 &\textbf{0.965} &0.961 &0.374 &0.744  \\ 
		& $50$ &0.790 &\textbf{1.000} &\textbf{1.000} &0.584 &0.947 \\
		& $80$ &0.938 &\textbf{1.000} &\textbf{1.000} &0.880 &0.998\\
		\hline
		$X \sim $ fBm & $n$& RPK  &  HSIC  &   DCOV   & $T_n^{(2,2)}$ & $T_n^{(2,\infty)}$   \\
		\hline 
		$Y=X^{2} +3\varepsilon$ & $30$ &\textbf{0.742} &0.728 &0.544 &0.732 &0.546 \\ 
		& $50$ &0.962 &0.946 &0.814 &0.883 &0.758  \\
		& $80$ &0.997 &0.998 &0.970 & 0.987 &0.922 \\
		\hline
		$Y=\sqrt{|X|}+\varepsilon $  & $30$ &\textbf{0.962} &0.925 &0.579 &0.580 &0.266\\ 
		& $50$ &\textbf{1.000} &0.999 &0.902 &0.830 &0.440 \\
		& $80$ &\textbf{1.000} &\textbf{1.000} &\textbf{1.000} &0.930 &0.680\\
		\hline 
		$Y=\varepsilon X+3 \varepsilon ' $  & $30$ &0.054 &0.109 &0.125 &\textbf{0.366} &0.246  \\ 
		& $50$ &0.053 &0.098 &0.131 &\textbf{0.586} &0.404\\
		& $80$ &0.062 &0.119 &0.132 &\textbf{0.804 }&0.634 \\
		\hline
		$Y \sim $ FOU  & $30$ &0.585 &0.394 &0.509 &0.460 &\textbf{0.806} \\ 
		& $50$ &0.760 &0.705 &0.801 &0.581 & \textbf{0.978}\\
		& $80$ &0.913 &0.956 &0.984 &0.707 & \textbf{1.000}\\
		\hline 
		$Y \sim$ FOU(2)  & $30$ &0.443 &0.847 &\textbf{0.909} &0.206 &0.426 \\ 
		& $50$ &0.665 &0.987 &\textbf{0.996} &0.326 &0.722 \\
		& $80$ &0.820 &\textbf{1.000} &\textbf{1.000} &0.542 &0.928 \\
		\hline
		\hline
		$X \sim$ OU($\lambda_1$)  & $30$ &0.509 &0.445 &0.462 &0.304 &\textbf{0.672} \\ 
		$Y \sim$ OU($\lambda_2$)  & $50$ &0.717 &0.777 &0.769 &0.448 &\textbf{0.888} \\ 
		& $80$ &0.889 &0.978 &0.965 &0.582 &\textbf{0.980} \\ 
		\hline
		$X \sim$ FOU($\lambda_1$)  & $30$ &0.305 &0.180 &0.175  &0.106 &\textbf{0.332} \\ 
		$Y \sim$ FOU($\lambda_2$)  & $50$ &0.475 &0.299 &0.313 &0.204 &\textbf{0.524} \\ 
		& $80$ &0.644 &0.596 &0.574 &0.210 &\textbf{0.738} \\ 
		\hline
	\end{tabular}
\end{table*}

\begin{table*}[h]
	\caption {$p$-values for the test between couples of the $Z$'s.}\label{pvalueZs}
	\begin{center}

	\begin{tabular}{|c|cccc|}
		\hline
		& $Z_2$  & $Z_3$  & $Z_4$  & $Z_5$ \\
		\hline 
		$Z_1$  &0.000  & 0.000 &0.109 &0.000 \\
		\hline
		$Z_2$  &  & 0.000 &0.394 &0.000 \\
		\hline
		$Z_3$  &   &  &0.373 &0.000 \\
		
		\hline
		$Z_4$  & & & &0.403 \\
		\hline
	\end{tabular}
	\end{center}             
\end{table*}

\section{Applications to Real Data}
In this section we will see a couple of applications to  meteorological data. Two applications to 
economic data can be found in 
\cite{Kalemke_implementa}.

\subsection{Temperature, humidity, wind and evaporation}
In this subsection we consider the meteorological data given in Table 7.2 of \cite{Rencher}. The data are 
$46$ observations grouped
into $11$ variables defined as follows: $Y_1=$``maximum daily air temperature'', $Y_2=$``minimum daily air temperature'', 
$Y_3=$``integrated area under daily air temperature curve'', $Y_4=$``maximum daily soil temperature'', $Y_5=$``minimum daily soil temperature'',
$Y_6=$``integrated area under daily soil temperature curve'', $Y_7=$``maximum daily relative humidity'', $Y_8=$``minimum daily relative humidity'',
$Y_9=$``integrated area under daily humidity curve'', $Y_{10}=$``total wind (in miles per day)'' and $Y_{11}=$``evaporation''. We consider the vectors
$Z_{1}=(Y_1,Y_2,Y_3)$, $Z_2=(Y_4,Y_5,Y_6)$, $Z_3=(Y_7,Y_8,Y_9)$ and the variables  $Z_4=Y_{10}$ and 
$Z_5=Y_{11}$. Taking into account 
 what was seen in the previous section, that there are no important differences between the use of 
 $T_n^{(1)}$,  $T_n^{(2)}$
 or $T_n^{(\infty)}$ as a test statistic,  we apply our independence test
between couples of $Z$'s using $T_n^{(2,2)}$ as the test statistic. In Table \ref{pvalueZs} we show 
the $p$-values of our test in each case.
In Figure \ref{rencher_dependogram} we show the dependogram of order $2$ of the mutual independence test of 
the $Z$'s, that is, the critical values at $5 \%$ and $10 \%$
and the value of our statistic. The approximate $p$-values and critical values were calculated under 
$m=1000$ replications by a permutation method
as has been suggested in \cite{Fernandez}. The test concludes that $Z_1, Z_2, Z_3$ and $ Z_5$  are pairwise
independent, but the wind ($Z_4$) does
not exhibit any dependence with any of the other variables. On the one hand, these conclusions are equivalent 
to those obtained in \cite{Bilodeau}. 
On the other hand, we observe that in the cases in which the test does not reject $H_0$, the difference between the observed and the critical values
is very large. The largest observed value, is  $0.2887$,  reached in the case $Z_1, Z_2$.

\begin{figure}[h!]
\begin{center}
\begin{tabular}{ll}
\includegraphics[scale=0.6]{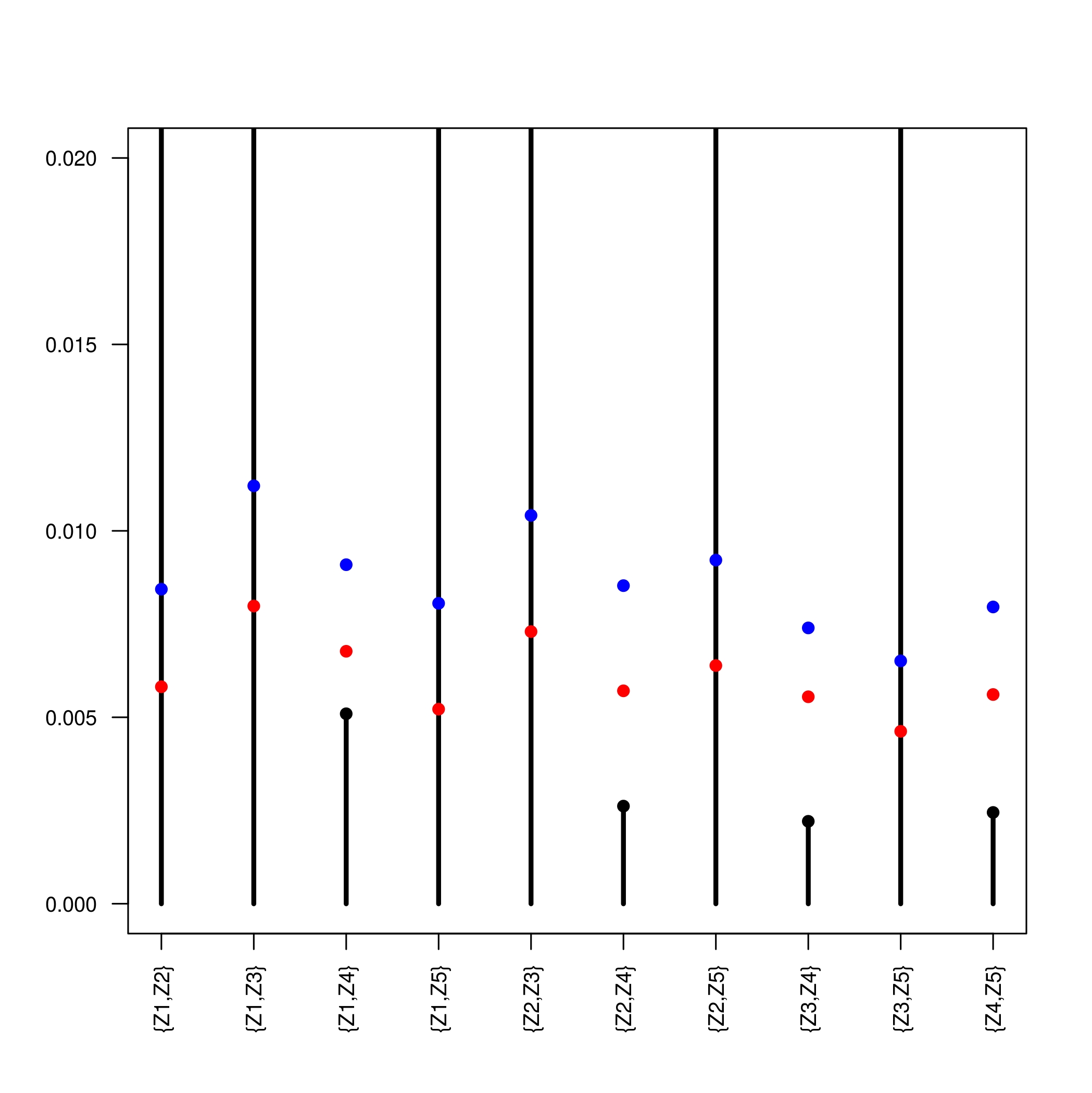}

\end{tabular}
\end{center}
\caption{Critical values at $5 \%$ (blue), $10 \%$ (red) and observed values (black) for the pairwise independence test between the $Z's$ variables.
}\label{rencher_dependogram}
\end{figure}

\subsection{Temperature, westbound wind, eastbound wind}
In this subsection we will consider the data set formed by the forecast temperature $(T)$, 
 westbound wind $(U)$ and eastbound wind $(V)$ at $850$ hPa (around $1200$m above sea level) from 
 each day from January 2012 to December 2012. There are a total of $341$ forecasts due to 
 the fact that 
 $25$ data points are missing. The numerical domain is shown in  Figure \ref{grilla} and consists of a total of $117\times 75=8775$ geographical points.
 The time horizon of the forecasts is 
 $24$ hours, and they are for 0:00 GMT hour of each day. The numerical simulations were obtained 
 using the WRF regional model \cite{Skamarock}, and the initial and lateral boundary conditions were
 obtained from the NCEP Global Forecast System, as in \cite{Cazes}.
 If we consider $(U_1,V_1,T_1), (U_2,V_2,T_2),...,(U_{341},V_{341},T_{341})$ where $U_i, V_i, T_i \in \mathbb{R}^{8775}$ for all $i=1,2,3,...,341$,
 the $p$-values for the independence test between $U$ and $V$ is equal to zero, and so on for the test
 between $U$ and $T$, and 
 $V$ and $T$. This is expected because for each geographical point $i$, the variables $U_i, V_i$ and $T_i$ 
 are pairwise dependent.
 Now we consider (for each day) every vector $U \in \mathbb{R}^{8775}$ to be decomposed
 as $U=(U_1,U_2,...,U_{75})$ where $U_i \in \mathbb{R}^{117}$.
 In this form, each $U_i$ represents the forecast of the $117$ geographical points at latitude $i$ and can be
 seen as a discretization of a curve
 at latitude $i$, ($U(i)$). Here, $i=1$ indicates the southeastern most latitude given in 
 Figure \ref{grilla} and $i=117$ the northeastern most latitude.  We consider the first 
 $30$ forecasts, corresponding to January 2012. In this way, we get a sample of $30$ curves for each 
 latitude $i$, 
 and we will test the mutual independence between $U_i$ and $U_j$ for $i=1,2,3,...,38$ and $j=76-i$.
 We decompose $T$ and $V$ analogously. It is to be expected that, at least for small values of $i$, the variables $U_i$ and $U_j$ would be independent, due 
 to the geographical distance, and  the same for the variables $V$ and $T$. In 
 Figures \ref{uvt}, \ref{uv}, \ref{ut} and \ref{vt} we show the dependograms
 for the independence test, using $T_n^{(2,1)}$ and $T_n^{(2,\infty)}$ statistics, between $U_i$ and 
 $U_{76-i}$ for each $i=1,2,...,38$ and the same for the variables $V$, $T$ and the other combinations between 
 $U,V$ and $T$.  In Figure \ref{uvt}, similar results between $T_n^{(2,1)}$ and $T_n^{(2,\infty)}$ are shown. 
 However, in the case of $U_i$ and $U_j$,
 $T_n^{(2,\infty)}$ detects the dependence in more cases than $T_n^{(2,1)}$. Both tests show that when $i$ and $j=76-i$ are close, then the variables
 $U_i$ and $U_j$ are dependent. The same occurs with $V_i$, $V_j$ and $T_i$, $T_j$. Also the geographical 
 region in which the vectors are 
 dependent is longer for $T$ than for $U$ and $V$. Figure \ref{uv} shows that $T_n^{(2,1)}$ performs better 
 than $T_n^{(2,\infty)}$ because for 
 $i \geq 32$ the test based on $T_n^{(2,1)}$ detects a dependence for both cases: $U_i, V_j$ and $V_i, U_j$.
 Figures \ref{ut} and \ref{vt} show that
 the tests based on  $T_n^{(2,1)}$ and $T_n^{(2,\infty)}$ perform similarly. 
 Also, still being geographically close, 
 the vectors $U_i$ and $T_j$ are independent.
 However both tests detect a dependence between $T:i$ and $U_j$ for $i=21$ to $i=27$ (Figure \ref{ut}). 
 Figure \ref{vt} shows that in most cases, $T_i$
 and $V_j$ are dependent, while for $V_i$ and $T_j$ the test does not detect a dependence except for the cases 
 in which $i$ and $j$ are close.

  \begin{figure}[h!]
\begin{center}
\begin{tabular}{ll}
\includegraphics[scale=0.20]{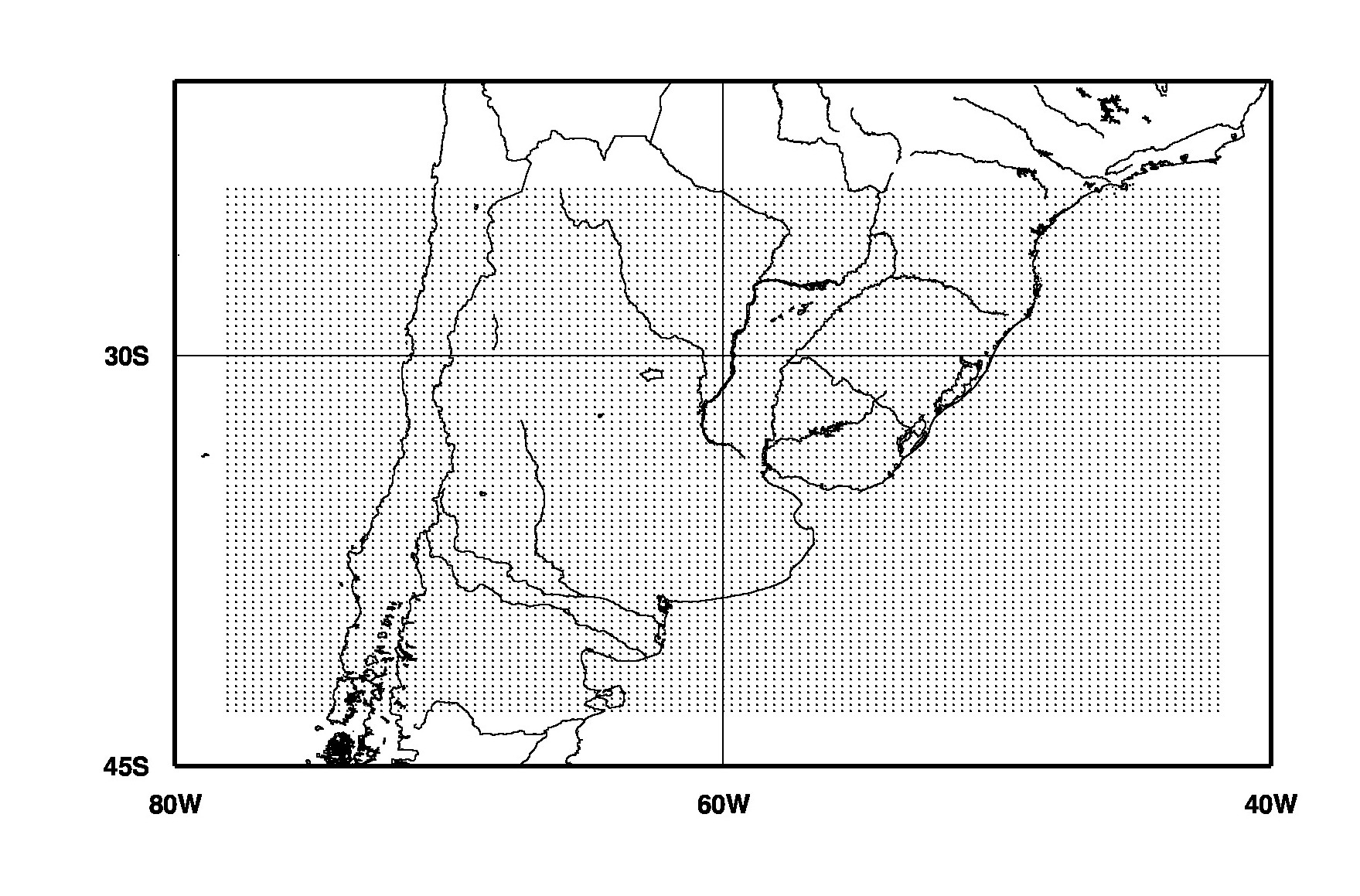}

\end{tabular}
\end{center}

\caption{$117\times 75=8775$ geographical points where the daily forecasts are  made.}
\label{grilla}
\end{figure}
 
 \begin{figure}[h!]
\begin{center}
\begin{tabular}{ll}
\includegraphics[scale=0.8]{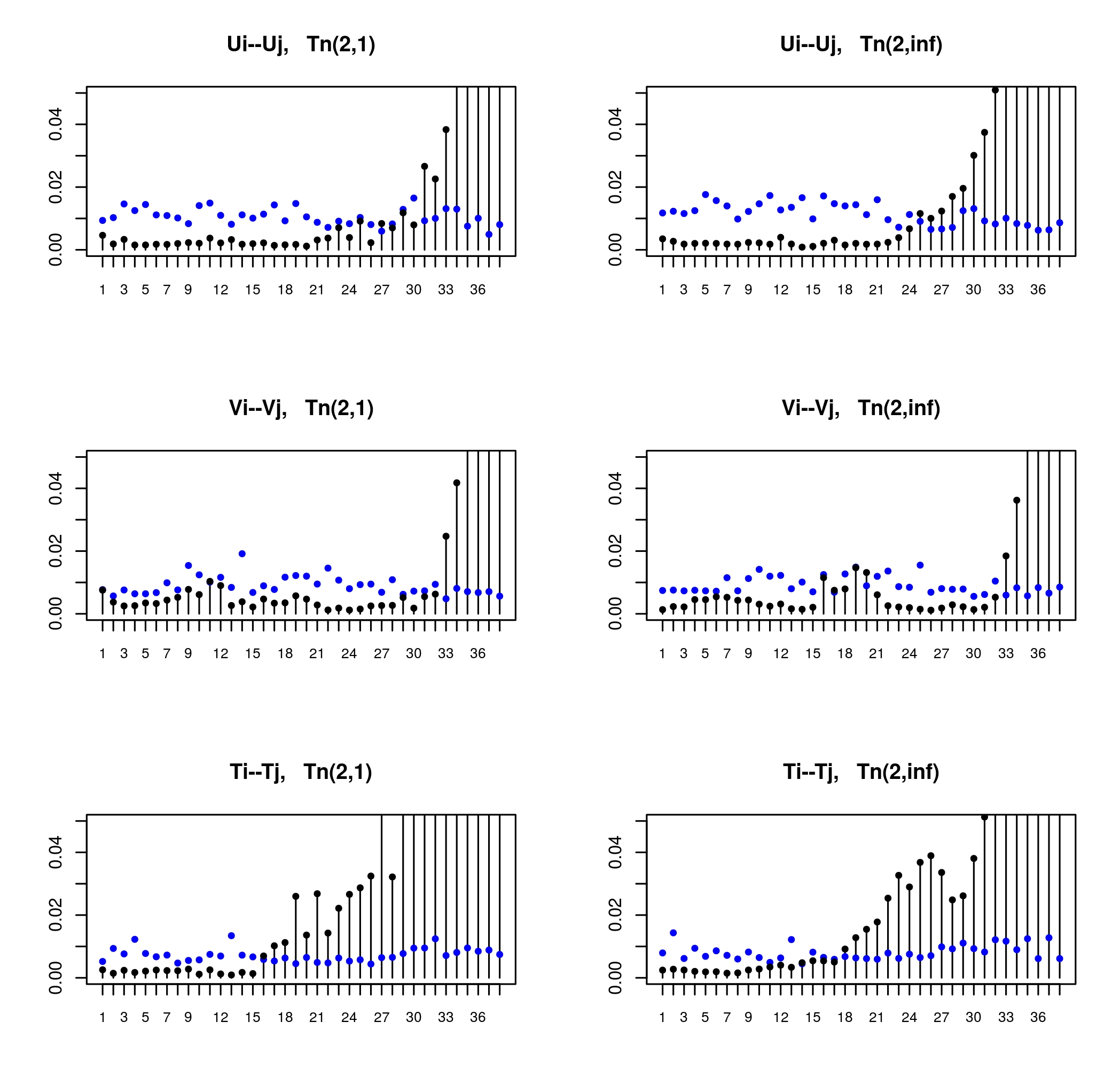}

\end{tabular}
\end{center}

\caption{Comparison between dependograms for $T_n^{(2,1)}$ and $T_n^{(2,\infty)}$ intra $U,V$ and $T$.
}\label{uvt}
\end{figure}

 \begin{figure}[h!]
\begin{center}
\begin{tabular}{ll}
\includegraphics[scale=0.8]{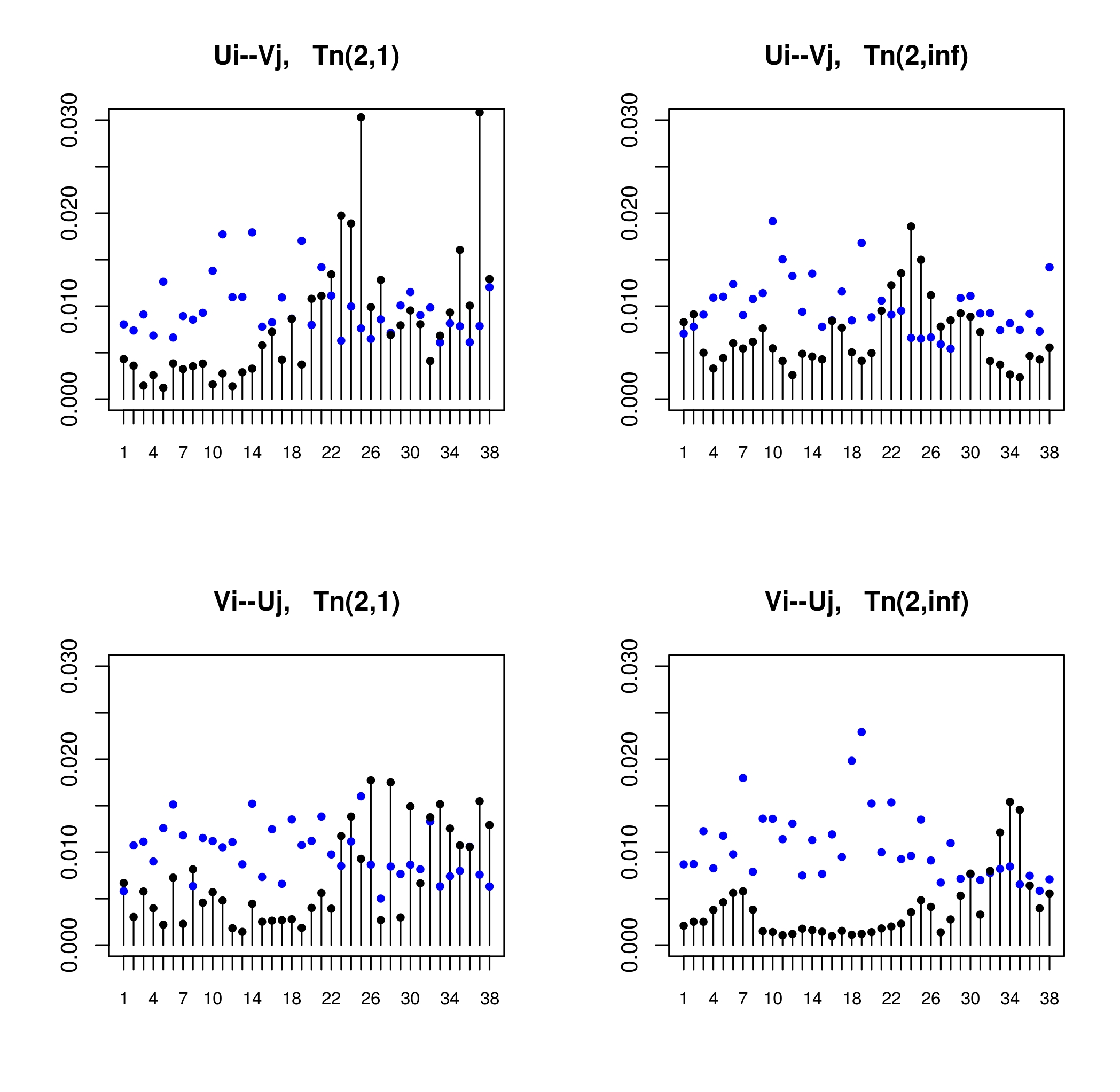}

\end{tabular}
\end{center}

\caption{Comparison between dependograms for $T_n^{(2,1)}$ and $T_n^{(2,\infty)}$ between $U$ and $V$.
}\label{uv}
\end{figure}

 \begin{figure}[h!]
\begin{center}
\begin{tabular}{ll}
\includegraphics[scale=0.8]{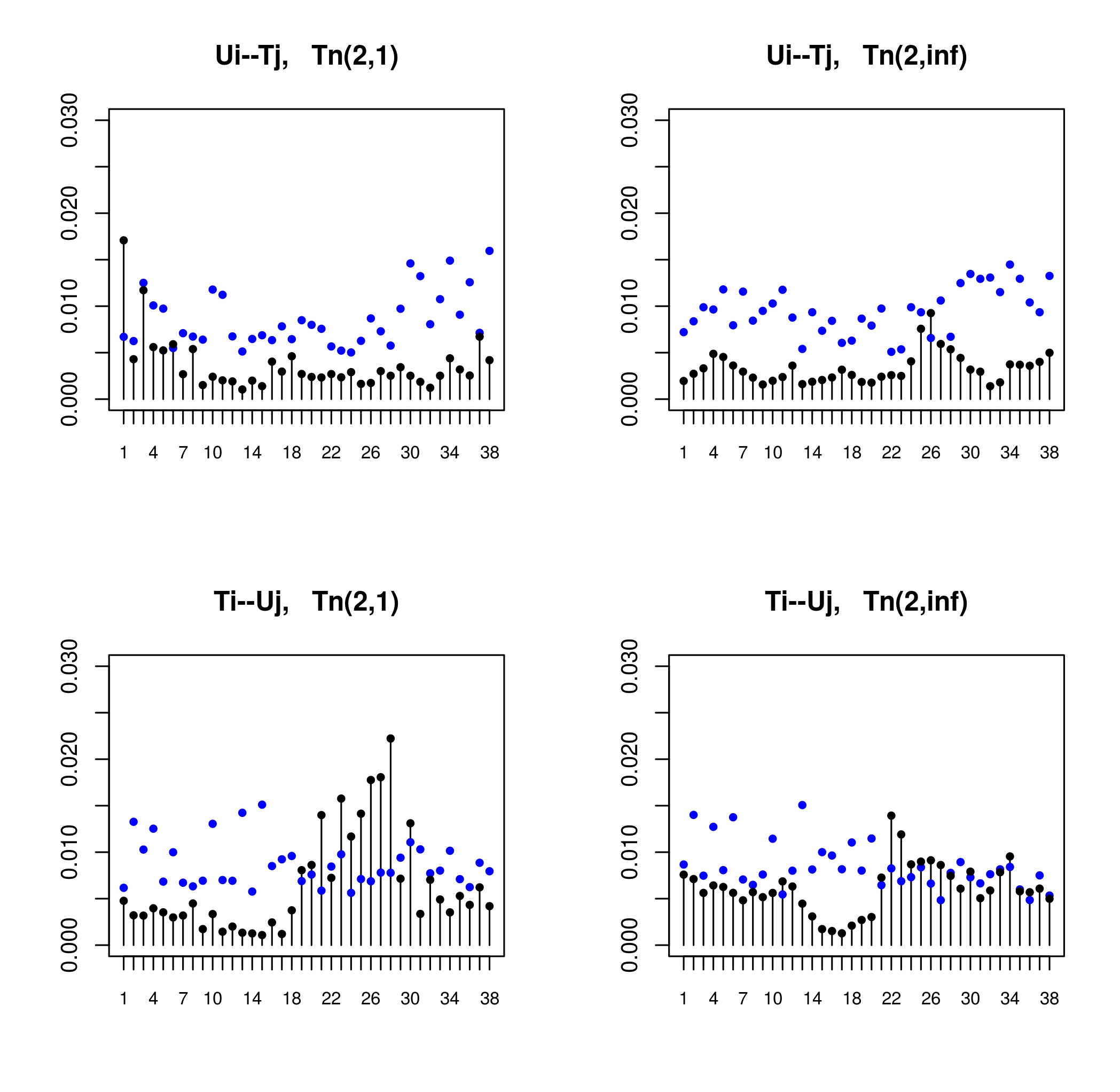}

\end{tabular}
\end{center}

\caption{Comparison between dependograms for $T_n^{(2,1)}$ and $T_n^{(2,\infty)}$ between $U$ and $T$.
}\label{ut}
\end{figure}

 \begin{figure}[h!]
\begin{center}
\begin{tabular}{ll}
\includegraphics[scale=0.8]{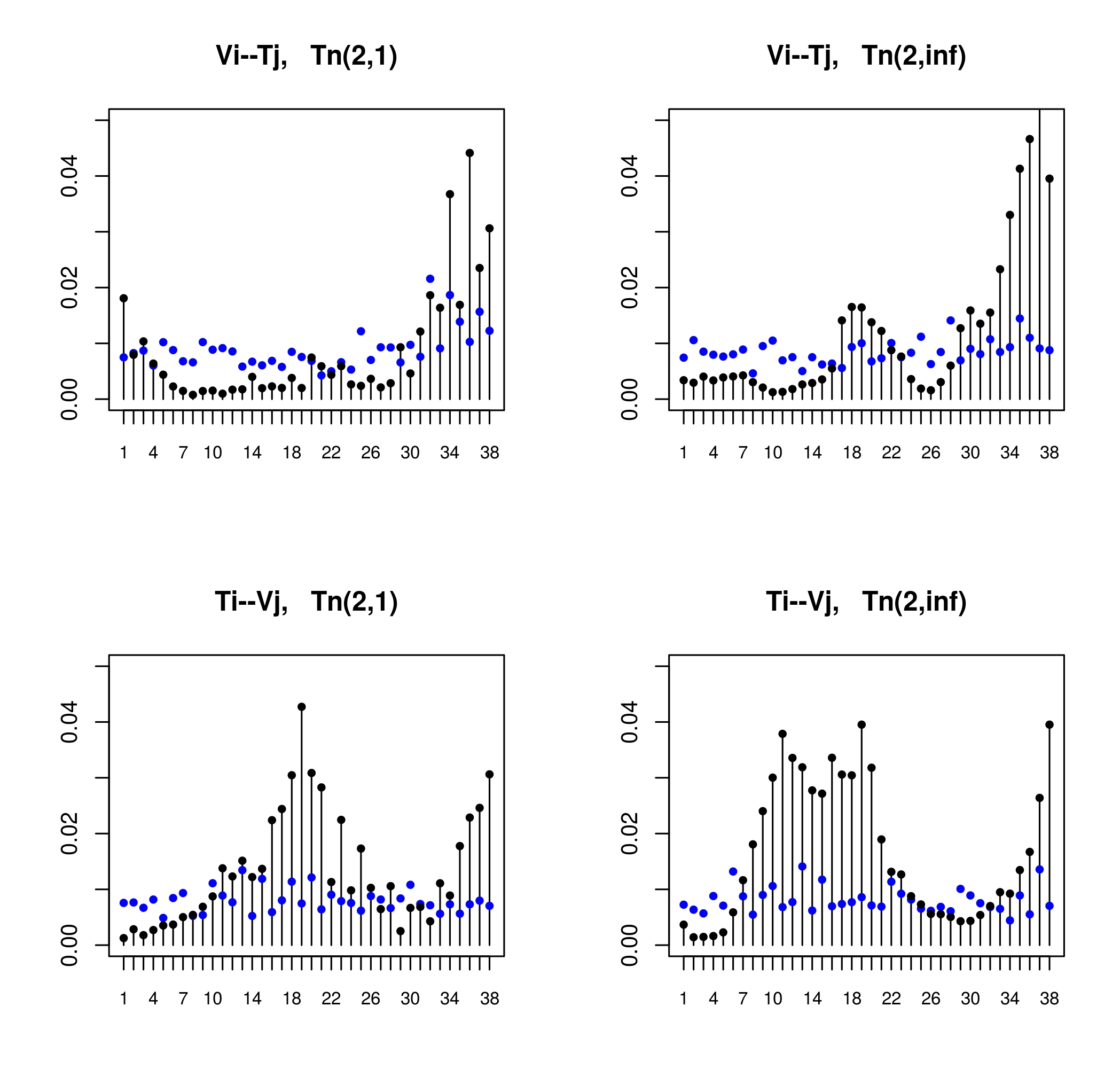}

\end{tabular}
\end{center}

\caption{Comparison between dependograms for $T_n^{(2,1)}$ and $T_n^{(2,\infty)}$ between $V$ and $T$.
}\label{vt}
\end{figure}

\section{Conclusions}
In this paper we have proposed several variants to perform the independence test
based on recurrence rates. We have shown how to calculate the test statistic in each one of
these cases.
When $X$ and $Y$ lie in high dimensional spaces, we have shown that the test
performs better as the distance function considered goes from the $L^{1} (l^{1})$ distance to 
the $L^{\infty} (l^{\infty})$ distance in some cases and in the opposite direction in other cases. Therefore,
the test statistic using the $L^{1} (l^{1})$ distance
and the $L^{\infty} (l^{\infty})$ distances to cover both possibilities can be suggested.
From simulations we obtain that in high dimension, our test clearly outperforms the 
competitors and widely used tests in almost all the alternatives considered. 
The test was performed on two data sets including small and large sample sizes and 
we have shown that in both cases the application of the test allows us to obtain interesting conclusions.
Taking this together with the simulations presented in \cite{Fernandez}, we can conclude that the 
independence test based on recurrence rates has very good performance for random variables, random
vectors, and also for random elements lying in high dimensional spaces.

\section{Proofs}

\begin{proof}[Proof of Proposition \ref{T2}]
\[\]
 The proof of this proposition is found in \cite{Fernandez}.
\end{proof}

\begin{proof}[Proof of Proposition \ref{T1}]
	\[ \]
	
 We re-order $d\left( X_{i},X_{j}\right) $ with $\left( i,j\right) \in I_{2}^{n}$
in the form $Z_{1},Z_{2},...,Z_{n}.$ Assume that $Z_{1}<Z_{2}<...<Z_{N},$
and we will use $T_{1},T_{2},...,T_{N}$ to denote the values of $d\left(
Y_{i},Y_{j}\right) $ using the same indexing. We also write $T_{1}^{\ast 
},T_{2}^{\ast },...,T_{N}^{\ast }$ for the order statistics of $T^{\prime }s.$

\begin{equation*}
\int_{0}^{+\infty }\int_{0}^{+\infty }\left\vert RR_{n}^{X,Y}\left(
r,s\right) -RR_{n}^{X}\left( r\right) RR_{n}^{Y}\left( s\right) \right\vert
g_{1}\left( r\right) g_{2}\left( s\right) drds=
\end{equation*}%
\begin{equation*}
\frac{1}{N}\int_{0}^{+\infty }g_{2}(s)ds\times \end{equation*}
\begin{equation*}\int_{0}^{+\infty }\left\vert
\sum_{i\neq j}\mathbf{1}_{\left\{ d\left( X_{i},X_{j}\right) <r,\text{ }%
d\left( Y_{i},Y_{j}\right) <s\right\} }-\frac{1}{N}\sum_{i\neq j}\mathbf{1}%
_{\left\{ d\left( X_{i},X_{j}\right) <r,\right\} }\sum_{h\neq k}\mathbf{1}%
_{\left\{ d\left( Y_{h},Y_{k}\right) <s\right\} }\right\vert g_{1}\left(
r\right) dr=
\end{equation*}%
\begin{equation}\label{formulaT1}
\frac{1}{N}\int_{0}^{+\infty }g_{2}(s)ds\int_{0}^{+\infty }\left\vert
\sum_{i=1}^{N}\mathbf{1}_{\left\{ Z_{i}<r,\text{ }T_{i}<s\right\} }-\frac{1}{%
N}\sum_{i=1}^{N}\mathbf{1}_{\left\{ Z_{i}<r\right\} }\sum_{j=1}^{N}\mathbf{1}%
_{\left\{ T_{j}<s\right\} }\right\vert g_{1}\left( r\right) dr.
\end{equation}%
Observe that 
\begin{equation*}
\int_{0}^{+\infty }\left\vert \sum_{i=1}^{N}\mathbf{1}_{\left\{ Z_{i}<r,%
\text{ }T_{i}<s\right\} }-\frac{1}{N}\sum_{i=1}^{N}\mathbf{1}_{\left\{
Z_{i}<r\right\} }\sum_{j=1}^{N}\mathbf{1}_{\left\{ T_{j}<s\right\}
}\right\vert g_{1}\left( r\right) dr=
\end{equation*}%
\begin{equation*}
\sum_{h=1}^{N-1}\int_{Z_{h}}^{Z_{h+1}}\left\vert \sum_{i=1}^{h}\mathbf{1}%
_{\left\{ T_{i}<s\right\} }-\frac{h}{N}\sum_{j=1}^{N}\mathbf{1}_{\left\{
T_{j}<s\right\} }\right\vert g_{1}\left( r\right) dr=
\end{equation*}%
\begin{equation*}
\sum_{h=1}^{N-1}\left\vert \sum_{i=1}^{h}\mathbf{1}_{\left\{ T_{i}<s\right\}
}-\frac{h}{N}\sum_{j=1}^{N}\mathbf{1}_{\left\{ T_{j}<s\right\} }\right\vert
\left( G_{1}\left( Z_{h+1}\right) -G_{1}\left( Z_{h}\right) \right) .
\end{equation*}

Then, (\ref{formulaT1}) is equal to 
\begin{equation*}
\frac{1}{N}\sum_{h=1}^{N-1}\left( G_{1}\left( Z_{h+1}\right) -G_{1}\left(
Z_{h}\right) \right) \int_{0}^{+\infty }\left\vert \sum_{i=1}^{h}\mathbf{1}%
_{\left\{ T_{i}<s\right\} }-\frac{h}{N}\sum_{j=1}^{N}\mathbf{1}_{\left\{
T_{j}<s\right\} }\right\vert g_{2}(s)ds=
\end{equation*}%
\begin{equation*}
\frac{1}{N}\sum_{h=1}^{N-1}\left( G_{1}\left( Z_{h+1}\right) -G_{1}\left(
Z_{h}\right) \right) \sum_{j=1}^{N-1}\int_{T_{j}^{\ast }}^{T_{j+1}^{\ast
}}\left\vert \sum_{i=1}^{h}\mathbf{1}_{\left\{ T_{i}<s\right\} }-\frac{h}{N}%
\sum_{j=1}^{N}\mathbf{1}_{\left\{ T_{j}<s\right\} }\right\vert g_{2}(s)ds=
\end{equation*}%
\begin{equation*}
\frac{1}{N}\sum_{h=1}^{N-1}\left( G_{1}\left( Z_{h+1}\right) -G_{1}\left(
Z_{h}\right) \right) \sum_{j=1}^{N-1}\int_{T_{j}^{\ast }}^{T_{j+1}^{\ast
}}\left\vert c(h,j)-\frac{jh}{N}\right\vert g_{2}(s)ds=
\end{equation*}%
\begin{equation*}
\frac{1}{N}\sum_{h,j=1}^{N-1}\left( G_{1}\left( Z_{h+1}\right) -G_{1}\left(
Z_{h}\right) \right) \left( G_{2}\left( T_{j+1}^{\ast }\right) -G_{2}\left(
T_{j}^{\ast }\right) \right) \left\vert c(h,j)-\frac{jh}{N}\right\vert ,
\end{equation*}%
where $c(h,j)=\sum_{i=1}^{h}\mathbf{1}_{\left\{ T_{i}<T_{j+1}^{\ast
}\right\} }$ is the number of elements of the vector $\left(
T_{1},T_{2},...,T_{h}\right) $ that are less than $T_{j+1}^{\ast }$ for $%
h,j=1,2,3,...,N-1.$ Thus, 
\begin{equation*}
T_{n}^{\left( 1\right) }=\frac{\sqrt{n}}{N}\sum_{h,j=1}^{N-1}\left(
G_{1}\left( Z_{h+1}\right) -G_{1}\left( Z_{h}\right) \right) \left(
G_{2}\left( T_{j+1}^{\ast }\right) -G_{2}\left( T_{j}^{\ast }\right) \right)
\left\vert c(h,j)-\frac{jh}{N}\right\vert .
\end{equation*}%
\end{proof}

\begin{proof}[Proof of Proposition \ref{Tinf}]
\[ \]
In accordance with Steps 1 and 2, we put $N=n(n-1)$ and re-order $\left\{ d\left( X_{i},X_{j}\right) \right\} _{i\neq j}$ as $%
Z_{1},Z_{2},...,Z_{N}$ such that $Z_{1}<Z_{2}<...<Z_{N}$ and $\left\{
d\left( Y_{i},Y_{j}\right) \right\} _{i\neq j}$ as $T_{1},T_{2},...,T_{N}$
maintaining the same indexing as $Z^{\prime }s$ (that is, if $d\left(
X_{i},X_{j}\right) =Z_{h}$, then $d\left( Y_{i},Y_{j}\right) =T_{h}$).
Observe that, to compute $T_n^{(\infty)}(r,s)$ for all $r,s>0$ it is enough to compute
$T_n^{(\infty)}(Z_{i},T_j^{*})$ for every $i,j=1,2,...,N.$ Then, the result follows immediately 
from  Steps 4 and 5.
 
\end{proof}

\textbf{Acknowledgements}
We would like to express our gratitude to Leonardo Moreno for his explanations about the independence test based on 
random projections and for giving us the  R code to use it, and to Gabriel Cazes for giving us the data
sets and the map of Figure \ref{grilla}.

\end{document}